\documentclass[final]{llncs}
\pdfoutput=1
\usepackage{xifthen}

\usepackage{amsmath,amsfonts}
\usepackage{color}
\usepackage{txfonts}
\usepackage{mathtools}
\usepackage{cleveref}
\usepackage[ruled, vlined,linesnumbered]{algorithm2e}
\usepackage{paralist}
\usepackage{booktabs}
\usepackage{array}
\usepackage{pgfplots}
\usepackage[utf8]{inputenc}
\usepackage{wrapfig}
\usepackage{graphicx}

 \newcommand{\theoremlike}[2]{\par\medskip\penalty-250%
{{\bfseries\noindent
#2 \ref{#1}.}}\it}

\newcommand{\thmhelperpre}[2]{\theoremlike{#1}{#2}}
\newcommand{\thmhelperpost}{\par\medskip}

\newenvironment{reftheorem}[1]{\thmhelperpre{#1}{Theorem}}{\thmhelperpost}

\usepackage{todonotes}

\usepackage{stmaryrd}

\usepackage{tikz}
\usetikzlibrary{myautomata}
\tikzstyle{max}=[shape=rectangle,draw,inner sep=0pt,minimum size=6mm,thick]
\tikzstyle{ran}=[shape=circle,draw,inner sep=0pt,minimum size=6mm,thick]

\newcommand{\optmdp}[2][]{
  \ifthenelse{\isempty{#1}}%
    {#2}     
    {#2_{#1}}
}

\newcommand{\Nset}{\mathbb{N}}
\newcommand{\extNset}{\overline{\Nset}}

\renewcommand{\vec}[1]{\mathbf{#1}}
\newcommand{\veccomp}[2]{#1(#2)}

\newcommand{\dist}{\mathcal{D}}
\newcommand{\mdp}{\mathcal{M}}
\newcommand{\states}{S}
\newcommand{\actions}{A}
\newcommand{\trans}{\Delta}
\newcommand{\cons}{C}
\newcommand{\reloads}{\mathit{R}}
\newcommand{\Ca}{\mathit{cap}}
\newcommand{\Succ}{\mathit{Succ}}
\newcommand{\target}{T}
\newcommand{\trunc}[2][]{\optmdp[#1]{\llbracket \, #2 \, \rrbracket}}
\newcommand{\strunc}[2][]{\optmdp[#1]{\llfloor \,#2 \,\rrfloor}}

\providecommand{\path}{}
\renewcommand{\path}{\alpha}
\newcommand{\fpath}{\alpha}
\newcommand{\run}{\mathit{\varrho}}
\newcommand{\hist}{\fpath}
\newcommand{\histpr}{\beta}
\newcommand{\histconc}{\odot}
\newcommand{\Runs}[1][]{\optmdp[#1]{\mathsf{Runs}}}
\newcommand{\histories}[1][]{\optmdp[#1]{\mathit{hist}}}
\newcommand{\pref}[1]{_{..#1}}
\newcommand{\suff}[1]{_{#1..}}
\newcommand{\infix}[2]{_{#1..#2}}
\newcommand{\rstate}[2][\run]{#1_{#2}}
\newcommand{\ract}[2][\run]{\mathit{Act}^{#2}(#1)}
\newcommand{\last}[1]{\mathit{last(#1)}}

\newcommand{\compatible}[3][]{\optmdp[#1]{\mathsf{Comp}}(#2,#3)}
\newcommand{\memstruct}{\mu}
\newcommand{\mem}{M}
\newcommand{\nextf}{\mathit{nxt}}
\newcommand{\upf}{\mathit{up}}
\newcommand{\meminit}{m_0}
\newcommand{\upfstep}{\upf^*}
\newcommand{\selector}{\Sigma}
\newcommand{\selrule}{\varphi}
\newcommand{\dom}{\mathit{dom}}
\newcommand{\srules}[1][]{\optmdp[#1]{\mathit{Rules}}}

\newcommand{\SafeRuns}{\mathsf{SafeRuns}}
\newcommand{\ReachRuns}[1][\target]{\mathsf{Reach}_{#1}}
\newcommand{\BuchiRuns}[1][\target]{\mathsf{B\ddot{u}chi}_{#1}}
\newcommand{\Safe}{\mathit{Safe}}
\newcommand{\SafePosReach}[1][\target]{\mathit{SafePR}_{#1}}
\newcommand{\SafeBuchi}[1][\target]{\mathit{SafeB\ddot{u}chi}_{#1}}

\newcommand{\MinInitCons}{\mathit{MinInitCons}}
\newcommand{\pathcons}{\mathit{cons}}
\newcommand{\reachcons}[1]{\mathit{ReachCons}_{#1}}
\newcommand{\minreach}[1]{\mathit{MinReach}_{#1}}

\newcommand{\SPRval}[1][]{\optmdp[#1]{\mathit{SPR\textnormal{-}Val}}}

\newcommand{\calO}{\mathcal{O}}

\newcommand{\suf}[1]{_{#1..}}

\newcommand{\probm}[2]{\mathbb{P}^{#1}_{#2}}

\newcommand{\cf}[1]{\mathcal{C}}  %

\newcommand{\len}[1]{len(#1)}

\newcommand{\safe}{\mathit{Safe}}
\newcommand{\infvec}{\boldsymbol{\infty}}
\newcommand{\relvar}{\mathit{Rel}}
\newcommand{\old}{\mathit{old}}

\newcommand{\varToRemove}{\mathit{ToRemove}}

\newcommand{\enlev}[3][]{\mathit{RL}_{#2}^{#1}(#3)}
\newcommand{\ifApp}[2]%
{\ifthenelse{\isundefined{\showappendix}}{#2}{#1}}

\begin{document}
	
\title{Qualitative Controller Synthesis\\for Consumption Markov Decision Processes\thanks{This work was partially supported by NASA under Early Stage Innovations grant No. 80NSSC19K0209, and by DARPA under grant No. HR001120C0065. Petr Novotný is supported by the Czech Science Foundation grant No. GJ19-15134Y}}

\author{František Blahoudek\inst{1} \and
	Tomáš Brázdil\inst{2} \and
	Petr Novotný\inst{2} \and
	Melkior Ornik\inst{3} \and
	Pranay Thangeda\inst{3} \and
	Ufuk Topcu\inst{1}}

\institute{Dept of Aerospace Engineering, The University of Texas at Austin, USA  \\
	\email{fandikb@gmail.com, utopcu@utexas.edu} \and
	Faculty of Informatics, Masaryk University, Brno, Czech Republic\\ \email{xbrazdil@fi.muni.cz, petr.novotny@fi.muni.cz} \and
	Dept of Aerospace Engineering, University of Illinois at Urbana-Champaign, Urbana, USA\\
	\email{mornik@illinois.edu, pranayt2@illinois.edu}}

\maketitle

\begin{abstract}
Consumption Markov Decision Processes (CMDPs) are probabilistic decision-making models of resource-constrained systems. In a CMDP, the controller possesses a certain amount of a critical resource, such as electric power. Each action of the controller can consume some amount of the resource. Resource replenishment is only possible in special \emph{reload states,} in which the resource level can be reloaded up to the full capacity of the system. The task of the controller is to prevent resource exhaustion, i.e. ensure that the available amount of the resource stays non-negative, while ensuring an additional linear-time property. We study the complexity of strategy synthesis in consumption MDPs with almost-sure Büchi objectives. We show that the problem can be solved in polynomial time. We implement our algorithm and show that it can efficiently solve CMDPs modelling real-world scenarios.
\end{abstract}


\section{Introduction}

In the context of formal methods, controller synthesis typically boils down to computing a strategy in an \emph{agent-environment} model, a nondeterministic state-transition model where some of the nondeterministic choices are resolved by the controller and some by an uncontrollable environment. Such models are typically either two-player graph games with an adversarial environment or Markov decision process (MDPs); the latter case being apt for modelling statistically predictable environments. In this paper, we consider controller synthesis for \emph{resource-constrained MDPs}, where the computed controller must ensure, in addition to satisfying some linear-time property, that the system's operation is not compromised by a lack of necessary resources.

\paragraph{Resource-Constrained Probabilistic Systems.} \emph{Resource-con\-strained} systems need a supply of some resource (e.g. power) for steady operation: the interruption of the supply can lead to undesirable consequences and has to be avoided. For instance, an autonomous system, e.g. an autonomous electric vehicle (\emph{AEV}), is not able to draw power directly from an endless source. Instead, it has to rely on an internal storage of the resource, e.g. a battery, which has to be replenished in regular intervals to prevent resource exhaustion. 
Practical examples of AEVs include driverless cars, drones, or planetary rovers~\cite{balaram2018mars}. In these domains, resource failures may cause a costly mission failure and even safety risks. Moreover, the operation of autonomous systems is subject to probabilistic uncertainty~\cite{SB:book}. Hence, in this paper, we study the resource-constrained strategy synthesis problem for MDPs.

\paragraph{Models of Resource-Constrained Systems \& Limitations of Current Approaches.} There is a substantial body of work on verification of resource-constrained systems~\cite{CdAHS:resource-interfaces,BFLMS:weak-upper-bound,BHR:battery,BBFLMR:energy-controller-synthesis,WHLK:energy-aware-scheduling,SSDNLB:energy-validation,FL:featured-weighted-automata,FZ:cost-parity-games,BDKL:energy-utility-probabilistic-mc,BDDKK:energy-utility-quantiles}. The typical approach is to model them as finite-state systems augmented with an integer-valued counter representing the current \emph{resource level,} i.e. the amount of the resource present in the internal storage. The resource constraint requires that the resource level never drops below zero.\footnote{In some literature, the level is required to stay positive as opposed to non-negative, but this is only a matter of definition: both approaches are equivalent.} In the well-known \emph{energy} model~\cite{CdAHS:resource-interfaces,BFLMS:weak-upper-bound}, each transition is labelled by an integer, and performing an $ \ell $-labelled transition results in $ \ell $ being added to the counter. Thus, negative numbers stand for resource consumption while positive ones represent re-charging by the respective amount. Many variants of both MDP and game-based energy models were studied, as detailed in the related work. In particular,~\cite{CHD:energy-MDPs} considers controller synthesis for energy MDPs with qualitative B\"uchi and parity objectives. The main limitation of energy-based agent-environment models is that in general, they are not known to admit polynomial-time controller synthesis algorithms. Indeed, already the simplest problem, deciding whether a non-negative energy can be maintained in a two-player energy game, is at least as hard as solving mean-payoff graph games~\cite{BFLMS:weak-upper-bound}; the complexity of the latter being a well-known open problem~\cite{Jurdzinski:parity-to-mp}. This hardness translates also to MDPs~\cite{CHD:energy-MDPs}, making polynomial-time controller synthesis for energy MDPs impossible without a theoretical breakthrough.

\emph{Consumption models,} introduced in~\cite{BCKN:consumption-games}, offer an alternative to energy models. In a consumption model, a non-negative integer, $ \Ca $, represents the maximal amount of the resource the system can hold, e.g. the battery capacity. Each transition is labelled by a non-negative number representing the amount of the resource \emph{consumed} when taking the transition (i.e., taking an $ \ell $-labelled transition decreases the resource level by $ \ell $). The resource replenishment is different from the energy approach. The consumption approach relies on the fact that reloads are often \emph{atomic events}, e.g. an AEV plugging into a charging station and waiting to finish the charging cycle. Hence, some states in the consumption model are designated as \emph{reload states,} and whenever the system visits a reload state, the resource level is replenished to the full capacity $ \Ca $. Modelling reloads as atomic events is natural and even advantageous: consumption models typically admit more efficient analysis than energy models~\cite{BCKN:consumption-games,Klaska:BT}. 
However, consumption models have not yet been considered in the probabilistic setting.

\paragraph{Our Contribution.} We study strategy synthesis in consumption MDPs with B\"uchi objectives. Our main theoretical result is stated in the following theorem.

\begin{theorem}
\label{thm:intro-main}
Given a consumption MDP $ \mdp $ with a capacity $ \Ca $, an initial resource level $ 0\leq d \leq \Ca$, and a set $ T $ of accepting states, we can decide, in polynomial time, whether there exists a strategy $ \sigma $ such that when playing according to $ \sigma $, the following \emph{consumption-B\"uchi objectives} are satisfied:
\begin{itemize}
\item Starting with resource level $ d $, the resource level never\footnote{In our model, this is equivalent to requiring that with probability 1, the resource level never drops below $ 0 $.} drops below $ 0 $.
\item With probability $ 1 $, the system visits some state in $ T $ infinitely often.
\end{itemize}
Moreover, if such a strategy exists then we can compute, in polynomial time, its poly\-no\-mial-size representation.
\end{theorem}

For the sake of clarity, we restrict to proving \Cref{thm:intro-main} for a natural sub-class of MDPs called \emph{decreasing consumption MDPs,} where there are no cycles of zero consumption. The restriction is natural (since in typical resource-constrained systems, each action -- even idling -- consumes some energy, so zero cycles are unlikely) and greatly simplifies presentation. In addition to the theoretical analysis, we implemented the algorithm behind~\Cref{thm:intro-main} and evaluated it on several benchmarks, including a realistic model of an AEV navigating the streets of Manhattan. The experiments show that our algorithm is able to efficiently solve large CMDPs, offering a good scalability. 

\paragraph{Significance.}
Some comments on \Cref{thm:intro-main} are in order. First, all the numbers in the MDP, and in particular the capacity $ \Ca $, are encoded in binary. Hence, ``polynomial time'' means time polynomial in the encoding size of the MDP itself and in $ \log(\Ca) $. In particular, a naive ``unfolding'' of the MDP, i.e. encoding the resource levels between $ 0 $ and $ \Ca $ into the states, does not yield a polynomial-time algorithm, but an exponential-time one, since the unfolded MDP has size proportional to $ \Ca $. We employ a value-iteration-like algorithm to compute minimal energy levels with which one can achieve the consumption-B\"uchi objectives.

A similar concern applies to the ``polynomial-size representation'' of the strategy $ \sigma $. To satisfy a consumption-B\"uchi objective, $ \sigma $ generally needs to keep track of the current resource level. Hence, under the standard notion of a finite-memory (FM) strategy (which views FM strategies as transducers), $ \sigma $ would require memory proportional to $ \Ca $, i.e. a memory exponentially large w.r.t. size of the input. However, we show that for each state $ s $ we can partition the integer interval $[0,\ldots,\Ca]  $ into polynomially many sub-intervals $ I_1^s,\ldots,I_k^s $ such that, for each $ 1\leq j \leq k $, the strategy $ \sigma $ picks the same action whenever the current state is $ s $ and the current resource level is in $ I_j^s $. As such, the endpoints of the intervals are the only extra knowledge required to represent $ \sigma $, a representation which we call a \emph{counter selector}. 
We instrument our main algorithm so as to compute, in polynomial time, a polynomial-size counter selector representing the witness strategy $ \sigma $.

Finally, we consider linear-time properties encoded by B\"uchi objectives over the states of the MDP. In essence, we assume that the translation of the specification to the B\"uchi automaton and its product with the original MDP model of the system were already performed. Probabilistic analysis typically requires the use of deterministic B\"uchi automata, which cannot express all linear-time properties. However, in this paper we consider qualitative analysis, which can be performed using restricted versions of non-deterministic B\"uchi automata that are still powerful enough to express all $ \omega $-regular languages. Examples of such automata are limit-deterministic B\"uchi automata~\cite{SEJK:ldba} or good-for-MDPs automata~\cite{HPSTSW:goodformdps}. Alternatively, consumption MDPs with parity objectives could be reduced to consumption-B\"uchi MPDs using the standard parity-to-B\"uchi MDP construction~\cite{Chaterjee:thesis,Alfaro:thesis,CY95,CHJH:quant-parity}. We abstract from these aspects and focus on the technical core of our problem, solving consumption-B\"uchi MDPs.

Consequently, to our best knowledge, we present the first polynomial-time algorithm for controller synthesis in resource-constrained MDPs with $ \omega $-regular objectives.


\paragraph{Related Work.} There is an enormous body of work on energy models. Stemming from the models introduced in~\cite{CdAHS:resource-interfaces,BFLMS:weak-upper-bound}, the subsequent work covered energy games with various combinations of objectives~\cite{CD:energy-parity-journal,BMRLL:average-energy-games,LLZ:limit-consumption,BHMRZ:bounding-average-energy-games,BHRR:energy-mp-games,BCDGR:energy-games,BCR:energy-MP-timed-games,BFLM:timed-observers}, energy games with multiple resource types~\cite{FLLS11:EnGames,JLR:multiweighted-energy,CRR14,VCDHRR:multi-mp-energy,JLS:fixed-dim-energy,CHaloupka13,BJK:eVASS-games,CHDHR:multi-energy-mean-payoff} or the variants of the above in the MDP~\cite{BKN:energy-mp-ATVA,MSTW:energy-parity-MDPs}, infinite-state~\cite{AAHMKT14:infinite-state-energy-games}, or partially observable~\cite{DDGRT10:PO-energy-MP} settings. As argued previously, the controller synthesis within these models is at least as hard as solving mean-payoff games. The paper~\cite{CHKN:energy-games-polynomial} presents polynomial-time algorithms for non-stochastic energy games with special weight structures. Recently, an abstract algebraic perspective on energy models was presented in~\cite{CFL:algebraic-realtime-energy,EFLQ:algebraic-energy-I,EFLQ:algebraic-energy-II}.

Consumption systems were introduced in~\cite{BCKN:consumption-games} in the form of consumption games with multiple resource types. Minimizing mean-payoff in automata with consumption constraints was studied in~\cite{BKKN:consumption-payoff}. 

Our main result requires, as a technical sub-component, solving the \emph{resource-safety} (or just \emph{safety}) problem in consumption MDPs, i.e. computing a strategy which prevents resource exhaustion. The solution to this problem consists (in principle) of a Turing reduction to the problem of minimum cost reachability in two-player games with non-negative costs. The latter problem was studied in~\cite{KBBEGRZ:short-path-interdiction}, with an extension to arbitrary costs considered in~\cite{BGHM:shortest-path-games} (see also~\cite{FGR:quantitative-languages}). We present our own, conceptually simple, value-iteration-like algorithm for the problem, which is also used in our implementation.

Elements of resource-constrained optimization and minimum-cost reachability are also present in the line of work concerning \emph{energy-utility quantiles} in MDPs~\cite{BDDKK:energy-utility-quantiles,BDKL:energy-utility-probabilistic-mc,BDKKW:PMC-energy-utility-analysis,BCDKK:energy-utility-probabilistic,HBFKN:parameter-energy-synthesis}. In this setting, there is no reloading in the consumption- or energy-model sense, and the task is typically to minimize the total amount of the resource consumed while maximizing the probability that some other objective is satisfied.

\paragraph{Paper Organization \& Outline of Techniques} After the preliminaries~(\Cref{sec:prelims}), we present counter selectors in~\Cref{sec:counter}. The next three sections contain the three main steps of our analysis. In~\Cref{sec:safety}, we solve the safety problem in consumption MDPs. The technical core of our approach is presented in~\Cref{sec:posreach}, where we solve the problem of \emph{safe positive reachability}: finding a resource-safe strategy which ensures that the set $ T $ of accepting states is visited with positive probability. Solving consumption-B\"uchi MDPs then, in principle, consists of repeatedly applying a strategy for safe positive reachability of $ T $, ensuring that the strategy is ``re-started'' whenever the attempt to reach $ T $ fails. Details are given in~\Cref{sec:buchi}. Finally,~\Cref{sec:exp} presents our experiments. Due to space constraints, most technical proofs were moved to the appendix.

\section{Preliminaries}
\label{sec:prelims}

We denote by $\Nset$ the set of all non-negative integers and by
$\extNset$ the set $\Nset \cup \{\infty\}$. Given a set $I$ and a
vector $\vec{v}\in\extNset^{I}$ of integers indexed by $I$, we use
$\veccomp{\vec{v}}{i}$ to denote the $i$-component of $\vec{v}$. We
assume familiarity with basic notions of probability theory. In
particular, a \emph{probability distribution} on an at most countable
set $X$ is a function $f\colon X \rightarrow [0,1]$ s.t. $\sum_{x\in
X} f(x) = 1$. We use $\dist(X)$ to denote the set of all probability
distributions on $X$.

\begin{definition}[CMDP]
A \emph{consumption Markov decision process} (CMDP) is a tuple %
$\mdp = (\states, \actions, \trans, \cons, \reloads, \Ca)$ where
$\states$ is a finite set of  \emph{states}, %
$\actions$ is a finite set of \emph{actions}, %
$\trans\colon \states\times \actions \rightarrow \dist(\states)$ is a
total \emph{transition function}, %
$\cons\colon \states \times \actions \rightarrow \Nset$ is a total
\emph{consumption function}, %
$\reloads\subseteq \states$ is a set of \emph{reload states} where the
resource can be reloaded, and %
$ \Ca $  is a \emph{resource capacity}.%
\end{definition}
\noindent%
Given a set $\reloads'\subseteq \states$, we denote by $\mdp(R')$ the
CMDP obtained from $\mdp$ by changing the set of reloads to $R'$.
%
%
Given $s\in \states$ and $a\in \actions$, we denote by $\Succ(s,a)$
the set $\{t\mid \trans(s,a)(t)>0\}$. A \emph{path} is a (finite or
infinite) state-action sequence $\path=s_1a_1s_2a_2s_3\dots\in
(\states\times \actions)^\omega \cup (\states\cdot
\actions)^*\cdot\states$ such that $s_{i+1}\in \Succ(s_i,a_i)$ for all
$i$. We define $\rstate[\path]{i}=s_i$ and $\ract[\path]{i}=a_i$. We
use $\path\pref{i}$ for the finite prefix $s_1 a_1 s_2\dots s_i$ of
$\path$, we use $ \path\suff{i} $ for the suffix $ s_i a_i
s_{i+1}\dots $, and $\path\infix{i}{j}$ for the infix $s_ia_i\ldots
s_j$. The \emph{length}  of a path $\path$ is the number $\len{\path}$
of actions on $\path$ ($\len{\path}=\infty$ if $\path$ is infinite).

A finite path $ \fpath $ is \emph{simple} if no state appears more than once on $ \fpath $. A finite path is a \emph{cycle} if it starts and ends in the same state. A CMDP is \emph{decreasing} if for every simple cycle $ s_1 a_1s_2 \ldots a_{k-1}s_k $ there exists $ 1 \leq i < k $ such that  $ \cons(s_i,a_i) >0 $. Throughout this paper we consider only decreasing CMDPs. The only place where this assumption is used are the proofs of~\Cref{thm:safety-main} and  \Cref{thm:buchi}.

An infinite path is called a \emph{run}. We typically name runs by variants
of the symbol $\run$. The set of all runs in $\mdp$ is denoted $ \Runs[\mdp]
$ or simply $\Runs$ if $\mdp$ is clear from context. A finite path is
also called \emph{history}. The set of all possible histories of
$\mdp$ is $\histories[\mdp]$ or simply $\histories$. We denote by
$\last{\hist}$ the last state of a history $\hist$. Let $\hist$ be a
history with $\last{\hist}=s_1$ and $\histpr=s_1a_1s_2a_2\ldots$; we
define a \emph{joint path} as $\hist\histconc\histpr = \hist
a_1s_2a_2\ldots$.

A \emph{strategy} for $ \mdp $ is a function $ \sigma \colon
\histories[\mdp] \rightarrow \actions$ assigning to each history an
action to play. A strategy is \emph{memoryless} if $ \sigma(\hist) =
\sigma(\histpr) $ whenever $ \last{\hist}=\last{\histpr} $, i.e., when
the decision depends only on the current state. We do not consider randomized strategies in this paper, as they are non-necessary for qualitative
$\omega$-regular objectives on finite MDPs~\cite{Alfaro:thesis,CY95,CHJH:quant-parity}.

A computation of $\mdp$ under the control of a given
strategy $\sigma$ from some initial state $s\in\states$ creates a
path. The path starts with $s_1=s$. Assume that the current path is
$\path$ and let $s_i=\last{\path}$ (we say that $\mdp$ is currently in
$s_i$). Then the next action on the path is $a_{i}=\sigma(\alpha)$ and
the next state $s_{i+1}$ is chosen randomly according to $\trans(s_i,
a_{i})$. Repeating this process \emph{ad infinitum} yields an infinite
sample run $\run$. We say that a $\run$ is \emph{$\sigma$-compatible} if it can be produced using this process,
and \emph{$s$-initiated} if it starts in $ s $. We denote the set of all $ \sigma $-compatible $ s $-initiated runs by
$\compatible[\mdp]{\sigma}{s}$. 

We denote by $\probm{\sigma}{\mdp,s}(\mathsf{A})$ the probability that
a sample run from $\compatible[\mdp]{\sigma}{s}$ belongs to a given
measurable set of runs $ \mathsf{A} $ (the subscript $ \mdp $ is
dropped when $ \mdp $ is known from the context). For details on the 
formal construction of measurable sets of runs as well as the
probability measure $\probm{\sigma}{\mdp,s}$ see~\cite{Ash:book}.



\subsection{Resource: Consumption, Levels, and Objectives}%
We denote by $ \Ca(\mdp) $ the battery capacity in the MDP $\mdp$.
A resource is consumed along paths and can be reloaded in the reload states up to the full capacity. For a path $\path=s_1a_1s_2\ldots$ we define the consumption of $\path$ as $\pathcons(\path) = \sum_{i=1}^{\len{\alpha}}\cons(s_i, a_i)$ (since the consumption is non-negative, the sum is always well defined, though possibly diverging). Note that $ \pathcons $ does not consider reload states at all. To accurately track the remaining amount of the resource, we use the concept of a \emph{resource level}.

%

\begin{definition}[Resource level]\label{def:energy-level}%
Let $\mdp$ be a CMDP with a set of reload states $\reloads$, let $\hist$ be a
history, and let $0\leq d \leq \Ca(\mdp)$ be an integer called
\emph{initial load}. Then the \emph{energy level after $\hist$
initialized by $d$}, denoted by $ \enlev[\mdp]{d}{\hist} $ or simply
as $\enlev{d}{\hist}$, is defined inductively as follows: for a zero-length
history $s$ we have $ \enlev[\mdp]{d}{s} = d $. For a non-zero-length history $\hist = \histpr a t$ we denote $ c
= \cons({\last{\histpr}},{a}) $, and put%
\[
\enlev[\mdp]{d}{\hist} =
\begin{cases}
  \enlev[\mdp]{d}{\histpr} - c &
    \text{if } \last{\histpr} \not\in\reloads \text{ and }%
    c \leq \enlev[\mdp]{d}{\histpr} \neq \bot\\
  \Ca(\mdp) - c &
    \text{if } \last{\histpr}\in \reloads \text{ and }%
    c\leq \Ca(\mdp) \text{ and }%
    \enlev[\mdp]{d}{\histpr} \neq \bot\\
  \bot & \text{otherwise}
\end{cases}
\]
\end{definition}

Let $\hist$ be a history and let $f,l \geq 0$ that are the minimal and
maximal indices $i$ such that $\rstate[\hist]{i}\in\reloads$,
respectively. Following the inductive definition of $\enlev{d}{\hist}$
it is easy to see that if we have $\enlev{d}{\hist} \neq \bot$, then $\enlev{d}{\hist\pref{i}} = d - \pathcons(\hist\pref{i})$ holds
for all $i \leq f$ and $\enlev{d}{\hist} = \Ca(\mdp) -
\pathcons(\hist\suf{l})$. Further, for each history $\hist$ and $d$
such that $e=\enlev{d}{\hist}\neq\bot$, and each history $\histpr$
suitable for joining with $\hist$ it holds that
$\enlev{d}{\hist\histconc\histpr} = \enlev{e}{\histpr}$.

A run $\run$ is \emph{$d$-safe} if and only if the energy level
initialized by $d$ is a non-negative number for each finite prefix of $\rho$, i.e. if
for all $i > 0$ we have $\enlev{d}{\run\pref{i}} \neq \bot$. We say
that a run is safe if it is $\Ca(\mdp)$-safe. The next lemma follows immediately from the definition of an energy level.

\begin{lemma}
\label{lem:safe-path-extension}
Let $\run=s_1a_1s_2\ldots$ be a $d$-safe run for some $d$ and let $\hist$ be a history such that $\last{\hist}=s_1$. Then the run $\hist\histconc\run$ is $e$-safe if $\enlev{e}{\hist}\geq d$.
\end{lemma}


\paragraph{Objectives} An \emph{objective} is a set of runs. The objective $\SafeRuns(d)$ contains exactly
$d$-safe runs. Given a \emph{target set} $\target\subseteq \states$ and $i \in
\Nset$, we define $\ReachRuns^i = \{ \run \in \Runs \mid \run_j \in
\target \text{ for some } 1 \leq j \leq i + 1 \}$ to be the set of all runs
that reach some state from $\target$ within the first $i$ steps. We put
$\ReachRuns = \bigcup_{i \in \Nset} \ReachRuns^i$. Finally,
the set $\BuchiRuns = \{\run \in \Runs \mid \rstate{i} \in \target \text{
for infinitely many }i\in \Nset\}$.

\paragraph{Problems} We solve three main qualitative problems for CMDPs,
namely \emph{safety}, \emph{positive reachability}, and \emph{Büchi}.

Let us fix a state $s$ and a target set of states $ \target $. We say that a strategy is
\emph{$d$-safe in $s$} if $ \compatible{\sigma}{s} \subseteq
\SafeRuns(d)$. We say that $\sigma$ is \emph{$\target$-positive $d$-safe in
$s$} if it is $d$-safe in $s$ and $\probm{\sigma}{s}(\ReachRuns) >
0$, which means that there exists a run in $\compatible{\sigma}{s}$
that visits $\target$. Finally, we say that $\sigma$ is \emph{$\target$-Büchi
$d$-safe in a state $s$} if it is $d$-safe in $s$ and $
\probm{\sigma}{s}(\BuchiRuns) = 1$.

The vectors $ \Safe $, $ \SafePosReach $ (PR for ``positive
reachability''), and $ \SafeBuchi $ of type $\extNset^\states$
contain, for each $s\in \states$, the minimal $d$ such that there
exists a strategy that is $d$-safe in $s$, $\target$-positive $d$-safe in
$s$, and $\target$-Büchi $d$-safe in $s$, respectively, and $\infty$ if no
such strategy exists.

The problems we consider for a given CMDP are:
\begin{itemize}
\item \emph{Safety:} compute the vector $ \Safe $ and a strategy that
is $ \Safe(s) $-safe in every $ s \in \states $. \item \emph{Positive
reachability:} compute the vector $ \SafePosReach $ and a strategy
that is $ \target $-positive $ \SafePosReach(s) $-safe in every state $ s $.
\item \emph{Büchi:} compute $ \SafeBuchi $ and a strategy
that is $ \target $-Büchi $ \SafeBuchi(s) $-safe in every state $ s $.
\end{itemize}

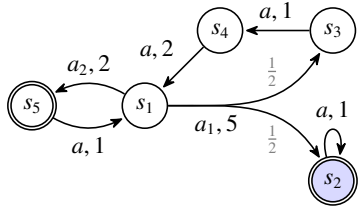
\begin{figure}[hbt]
\centering
\tikzstyle{target}=[fill=blue!15]
\begin{tikzpicture}[automaton]
\node[state] (start)    at (0,0)  {$s_1$};
\node[state,accepting,target] (t)   at (2.5,-1) {$s_2$};
\node[state,accepting] (rel)   at (-1.5,0) {$s_5$};
\node[state] (mis)      at (2.5,1)  {$s_3$};
\node[state] (pathhome) at (1,1)  {$s_4$};

\path[->,auto,swap]
(start) edge[bend right] node[] {$a_2, 2$} (rel)
(rel) edge[bend right] node[] {$a, 1$} (start)
(start) edge[out=0, in=120,looseness=1.1] node[pos=.25] {$a_1, 5$} node[prob,above,pos=.6] {$\frac{1}{2}$}(t)
(start) edge[out=0, in=240,looseness=1.1] node[prob,below,pos=.6] {$\frac{1}{2}$}(mis)
(mis) edge node {$a, 1$} (pathhome)
(pathhome) edge node {$a, 2$} (start)
(t) edge[loop above] node{$a, 1$} (t)
;
\end{tikzpicture}
\caption{CMDP $\mdp=(\{s_1, s_2, s_3, s_4, s_5\}, \{a_1, a_2\},
\Delta, \cons, \{s_2, s_5\}, 20)$ where distributions in $\Delta$ are
indicated by gray numbers (we leave out $\textcolor{gray}{1}$ when an
action has only one successor), and the cost of an action follows its
name in the edge labels. Actions labeled by $a_i$ represent that
$\Delta$ and $\cons$ are defined identically for both actions $a_1$ a
$a_2$. The blue background indicates a target set $\target=\{s_2\}$, while the double circles represent the reload states.}
\label{fig:example}
\end{figure}

We illustrate the key concepts using the example CMDP $\mdp$ in
\Cref{fig:example}. Consider the parameterized history $\hist^i =
(s_1a_2s_5a_2)^is_1$. Then $\pathcons(\hist^i) = 3i$ while
$\enlev{2}{\hist^i} = 19$ for all $i \geq 1$. Thus, a strategy, that always picks $a_2$ in $s_1$ is
$d$-safe in $s_1$ for all $d \geq 2$. On the other hand, a strategy
that always picks $a_1$ in $s_1$ is \emph{not} $d$-safe in $s_1$ for
any $0\leq d \leq 20$. Now consider again the strategy that always
picks $a_2$; such a strategy is $2$-safe in $s_1$, but is not useful
if we attempt to eventually reach $\target$. Hence memoryless strategies are not sufficient in our setting. Consider instead a strategy $ \sigma $ that, in $ s_1 $, picks $a_1$
whenever the current resource level is at least $ 10 $ and picks $ a_2 $ otherwise. Such a strategy is $ 2 $-safe in $ s_1 $ and guarantees reaching $ s_2 $ with a positive probability: we need at least 10 units of
energy to return to  $s_5$ in the case we are unlucky and picking $ a_1 $ leads us to $ s_3 $. If we are lucky, $ a_1 $ leads us to $s_2$ by consuming just $ 5 $ units of the resource, witnessing that $ \sigma $ is
$\target$-positive. As a matter of fact, during \emph{every} revisit of $ s_5 $ there is a $ \frac{1}{2} $ chance of hitting $ s_2 $ during the next try, so $ \sigma $ actually ensures that $ s_2 $ is visited with probability 1. 

We note that solving a CMDP is very different from solving a consumption 2-player game~\cite{BCKN:consumption-games}. Indeed, imagine that in Figure~\ref{fig:example}, the
outcome of the action $a_1$ from state $s_1$ is resolved by an
adversarial player. In such a game, there is no strategy that would
guarantee reaching $\target$ at all.

The strategy $ \sigma $ we discussed above uses finite memory
 to track the
resource level exactly. An efficient representation of such strategies
is described in the next section.

\section{Counter Strategies}%
\label{sec:counter}

In this section, we define a succinct representation of finite-memory strategies via so called counter selectors.
Under the standard definition, a strategy $\sigma$ is a \emph{finite
memory} strategy, if $\sigma$ can be encoded by a \emph{memory structure}, a
type of finite transducer. Formally, a memory structure is a tuple
$\memstruct=(\mem,\nextf,\upf,\meminit)$ where %
  $\mem$ is a finite set of \emph{memory elements}, %
  $\nextf\colon \mem \times \states \rightarrow \actions$ is a
  \emph{next action} function, %
  $\upf\colon \mem \times \states\times\actions\times
  \states\rightarrow \mem$ is a \emph{memory update} function, and %
  $\meminit\colon \states\rightarrow \mem$ is the \emph{memory
  initialization function}. %
The function $ \upf $ can be lifted to a function $
\upfstep\colon  \mem \times \histories \rightarrow \mem $ as follows.
\[
\upfstep(m,\hist) =
\begin{cases}
m & \text{if } \hist = s \text{ has length 0}\\
\upf\big(\upfstep(m,\histpr), \last{\histpr}, a, t\big) &
\text{if } \hist = \histpr a t
\text{ for some }
a\in\actions \text{ and } t\in\states
\\
\end{cases}
\]
The structure $\memstruct$ encodes a strategy $\sigma_\memstruct$ such
that for each history $\hist=s_1a_1s_2\ldots s_n$ we have
$\sigma_\memstruct(\alpha) = \nextf\big(\upfstep(\meminit(s_1),\hist),
s_n\big) $.

In our setting, strategies need to track energy levels of
histories. Let us fix an CMDP $\mdp = (\states, \actions, \trans,
\cons, \reloads, \Ca)$. A non-exhausted energy level is always a number between $0$
and $\Ca(\mdp)$, which can be represented with a binary-encoded
bounded counter. We call strategies with such counters \emph{finite
counter (FC) strategies}. An FC  strategy selects actions to
play according to \emph{selection rules}.

\begin{definition}[Selection rule]
A \emph{selection rule} $\selrule$ for $\mdp$ is a partial function
from the set $\{0,\ldots,\Ca(\mdp)\}$ to $A$. Undefined value for some
$n$ is indicated by $\selrule(n)=\bot$.
\end{definition}

We use $\dom(\selrule) = \{n \in \{0 ,\ldots,\Ca(\mdp)\} \mid
\selrule(n)\neq \bot\}$ to denote the domain of $\selrule$ and we use
$\srules[\mdp]$ or simply $\srules$ for the set of all selection rules
for $\mdp$. Intuitively, a selection according to rule $\selrule$
selects the action that corresponds to the largest value from
$\dom(\varphi)$ that is not larger than the current energy level. To
be more precise, if $ \dom(\selrule) $ consists of numbers $ n_1 < n_2
< \cdots < n_k $, then the action to be selected in a given moment is
$\selrule(n_i)$, where $n_i$ is the largest element of $
\dom(\selrule) $ which is less then or equal to the current amount of
the resource. In other words, $\selrule(n_i)$ is to be selected if the
current resource level is in $ [n_i,n_{i+1}) $ (putting $ n_{k+1} =
\infty $).

\begin{definition}[Counter selector]
A \emph{counter selector} for $\mdp$ is a function $ \selector\colon
\states \rightarrow \srules $.
\end{definition}

A counter selector itself is not enough to describe a strategy. A
strategy needs to keep track of the energy level throughout the path.
With a vector $\vec{r}\in \{0,\ldots,\Ca(\mdp)\}^{\states}$ of initial
resource levels, each counter selector $\selector$ defines a strategy
$\selector^\vec{r}$ that is encoded by the following memory structure
$ (\mem,\nextf,\upf,\meminit)$ with $a\in\actions$ being a globally
fixed action (for uniqueness). We stipulate that $\bot < n$ for all $n\in \Nset$.

\begin{compactitem}
\item $\mem = \{\bot\} \cup \{0,\ldots,\Ca(\mdp)\}$.%
\item Let $m\in \mem$ be a memory element, let $s\in \states$ be a
state, let $n\in\dom(\selector(s))$ be the largest element of
$\dom(\selector(s))$ such that $n \leq m$. Then
$\nextf(m,s) = \selector(s)(n)$ if $n$ exists, and $\nextf=a$ otherwise.

\item The function $ \upf $ is defined for each $ m\in \mem, a\in
\actions, s,t\in \states $ as follows.
\[
\upf(m,s,a,t) =
\begin{cases}
  m - \cons(s,a) &
    \text{if } s \not \in \reloads \text{ and }
    \cons(s,a) \leq m \neq \bot\\
  \Ca(\mdp) -\cons(s,a) &
    \text{if } s \in \reloads \text{ and }
    \cons(s,a) \leq \Ca(\mdp) \text{ and }
    m \neq \bot\\
\bot & \text{otherwise}.
\end{cases}
\]
\item The function $ \meminit $ is $ \meminit(s) = \veccomp{\vec{r}}{s}$.
\end{compactitem}

A strategy $ \sigma $ is a finite counter (FC) strategy if there is a
counter selector $ \selector $ and a vector $ \vec{r} $ such that
$\sigma = \selector^{\vec{r}}$. The counter selector can be imagined
as a finite-state device that implements $ \sigma $ using $ \calO(\log
(\Ca(\mdp))) $ bits of additional memory (counter) used to represent
numbers $0,1,\ldots,\Ca(\mdp)$. The device uses the counter to keep
track of the current resource level, the element $ \bot $ representing
energy exhaustion. Note that a counter selector can be exponentially more succinct than the corresponding memory structure.

%
\section{Safety}\label{sec:safety}%
In this section, we 
present an
algorithm that computes, for each state, the minimal value $d$ (if it 
exists) such that there exists a $d$-safe strategy from that state. We also provide
the corresponding strategy. In the remainder of the section we fix an MDP $ \mdp $.

A $d$-safe run has the following two properties: (i) It consumes at
most $d$ units of the resource (energy) before it reaches the first reload
state, and (ii) it never consumes more than $\Ca(\mdp)$ units of
the resource between 2 visits of reload states. To ensure (ii), we need to
identify a maximal subset $\reloads'\subseteq \reloads$ of reload states for which there is a strategy $\sigma$ that, starting in some $r\in
\reloads'$, can always reach $\reloads'$ again (within at least one
step) using at most $\Ca(\mdp)$ resource units. The $ d $-safe strategy we seek can be then assembled from  $\sigma$ and from a strategy that suitably navigates towards
$\reloads'$, which is needed for (i).

In the core of both properties (i) and (ii) lies the problem of \emph{minimum cost reachability.} Hence, in the next subsection, we start with presenting necessary results on this problem. 


\subsection{Minimum Cost Reachability}

The problem of minimum cost reachability with non-negative costs was studied before~\cite{KBBEGRZ:short-path-interdiction}. Here we present a simple approach to the problem used in our implementation.

\begin{definition}
Let $\target\subseteq \states$ be a set of \emph{target} states, let
$\path=s_1 a_1 s_2 \ldots$ be a finite or infinite path, and let $1
\leq f$ be the smallest index such that $s_f\in \target$. We define
\emph{consumption of $\path$ to $\target$} as
$\reachcons{\mdp,\target}(\path) = \pathcons(\alpha\pref{f})$
if $f$ exists and we set $\reachcons{\mdp,\target}(\path) = \infty$
otherwise. For a strategy $\sigma$ and a state $s\in \states$ we define
$\reachcons{\mdp,\target}(\sigma, s) =
\sup_{\run\in\compatible{\sigma}{s}} \reachcons{\mdp,\target}(\run)$.

\noindent
A \emph{minimum cost reachability of $\target$ from $s$} is a
vector defined as
\[
\minreach{\mdp, \target}(s) = \inf\big\{
\reachcons{\mdp,\target}(\sigma, s) \mid \sigma 
\text{ is a strategy for }\mdp
\big\}.
\]
\end{definition}

As usual, we drop the subscript $_\mdp$ when $\mdp$ is clear from
context. Intuitively, $d=\minreach{\target}(s)$ is the minimal initial
load with which some strategy can ensure reaching $ T $ with consumption at most $ d $, when starting in $ s $. We say that a
strategy $\sigma$ is optimal for $\minreach{\target}$ if we have that
$\minreach{\target}(s) = \reachcons{\target}(\sigma, s)$ for all
states $s\in \states$.

We also define functions $\reachcons{\mdp, \target}^+$ and the vector
$\minreach{\mdp,\target}^+$ in a similar fashion with one exception:
we require the index $f$ from definition of
$\reachcons{\mdp,\target}(\path)$ to be strictly larger than 1, which
enforces to take at least one step to reach $T$.

For the rest of this section, fix a target set $ T $ and consider the following functional $ \mathcal{F} $:

\[
\veccomp{\mathcal{F}(\vec{v})}{s} =
\begin{cases}
  \min_{a\in \actions} 
    \left(\cons(s,a) + \max_{t\in \Succ(s,a)} \veccomp{\vec{v}}{t}\right) &
    s \not \in \target\\
  0 & s \in \target

\end{cases}
\]

$ \mathcal{F} $ is a simple generalization of the standard Bellman functional used for computing shortest paths in graphs. The proof of the following Theorem is rather standard and is omitted for brevity.

\begin{theorem}
\label{thm:minreach-main} Denote by $n$ the length of the longest
simple path in $\mdp$. Let $\vec{x}_T$ be a vector such that
$\veccomp{\vec{x}_T}{s}=0$ if $s\in \target$ and
$\veccomp{\vec{x}_T}{s}=\infty$ otherwise. Then iterating
$\mathcal{F}$ on $\vec{x}_T$ yields a fixpoint in at most $n$ steps
and this fixpoint equals $\minreach{\target}$.
\end{theorem}

To compute $\minreach{\mdp, \target}^+$, we construct a new CMDP
$\widetilde\mdp$ from $\mdp$ by adding a copy $\tilde{s}$ of each
state $s \in \states$ such that dynamics in $\tilde{s}$ is the same as
in $s$; i.e. for each $a\in\actions$,
$\trans(\tilde{s},a)=\trans(s,a)$ and $\cons(\tilde{s},a) =
\cons(s,a)$. We denote the new state set as $\widetilde{\states}$.
We don't change the set of reload states, so $\tilde{s}$ is \emph{never} in $T$, even if $ s $ is. Given the new CMDP
$\widetilde\mdp$ and the new state set as $\widetilde{\states}$, the following lemma is straightforward.

\begin{lemma}
\label{lem:mininitcons-minreach}%
Let $\mdp$ be a CMDP and let $\widetilde{\mdp}$ be the CMDP
constructed as above. Then for each state $s$ of $\mdp$ it holds
$\minreach{\mdp, \target}^+(s) =
\minreach{\widetilde{\mdp},\target}(\tilde{s})$.
\end{lemma}

\subsection{Safely Reaching Reload States}%
In the following, we use $\MinInitCons_\mdp$ (read \emph{minimal
initial consumption}) for the vector $\minreach{\mdp, \reloads}^+$ --
minimal resource level that ensures we can surely reach a reload state in at
least one step. By Lemma~\ref{lem:mininitcons-minreach} and
Theorem~\ref{thm:minreach-main} we can construct $\widetilde{\mdp}$
and iterate the operator $\mathcal{F}$ for $|S|$ steps to compute
$\MinInitCons_\mdp$. Note that $S$ is the state space of $\mdp$ since
introducing the new states into $ \widetilde{\mdp} $ did not increase the length of the maximal
simple path. However, we can avoid the construction of
$\widetilde{\mdp}$ and still compute $\MinInitCons_\mdp$ using a \emph{truncated} version of the functional $ \mathcal{F} $, which is the approach used in our implementation.
%
We first introduce the following truncation operator:
\[
\veccomp{\strunc[\mdp]{\vec{x}}}{s} =
\begin{cases}
  \vec{x}(s) &
    \text{if } s \not \in \reloads, \\
  0 & \text{if }s \in \reloads.\\
\end{cases}\]
Then, we define a truncated functional $ \mathcal{G} $ as follows:
\[
\veccomp{\mathcal{G}(\vec{v})}{s} =
\min_{a\in \actions}
  \left(\cons(s,a) +
    \max_{s'\in \Succ(s,a)} \veccomp{\strunc{\vec{v}}_{\mdp}}{s'}\right).
\]

\begin{algorithm}[t]
\KwIn{CMDP $\mdp=(\states, \actions, \trans, \cons, \reloads, \Ca)$}
\KwOut{The vector $\MinInitCons_\mdp$}
initialize $\vec{x}\in \extNset^{\states}$ to be $\infty$ in every component\;
\Repeat{$\,\vec{x}_{\mathit{old}} = \vec{x}$}{
  $ \vec{x}_\old \leftarrow \vec{x} $\;
  \ForEach{$ s \in \states $}{
    $c \leftarrow
      \min_{a\in \actions}\Big\{ \cons(s,a) + \max_{s' \in \Succ(s,a)} \veccomp{\strunc{\vec{x}_\old}_{\mdp}{}}{s'}   \Big\}$\;
    \If{$c < \vec{x}(s)$}{
      $ \vec{x}(s) \leftarrow c$\;
    }
  }
}
\Return{$\vec{x}$}
\caption{Algorithm for computing $\MinInitCons_\mdp$.}
\label{algo:mininitcons_iterative}
\end{algorithm}

The following lemma connects the iteration of $ \mathcal{G} $ on $ \mdp $ with the iteration of $ \mathcal{F} $ on $ \widetilde{\mdp} $.

\begin{lemma}
\label{lem:minitconst-reach-to-direct} Let $\infvec\in
\extNset^{\states}$ be a vectors with all components equal to
$\infty$. Consider iterating $\mathcal{G}$ on $\infvec$ in $\mdp$ and
$\mathcal{F}$ on $ \vec{x}_\reloads $ in $ \widetilde{\mdp} $. Then
for each $i\geq 0$ and each $s\in\reloads$ we have
$\veccomp{\mathcal{G}^i(\infvec)}{s} =
\veccomp{\mathcal{F}^i(\vec{x}_\reloads)}{\tilde{s}}$ and for every
$s\in\states\setminus \reloads$ we have $
\veccomp{\mathcal{G}^i(\infvec)}{s} =
\veccomp{\mathcal{F}^i(\vec{x}_\reloads)}{s} $.%
\end{lemma}

Algorithm~\ref{algo:mininitcons_iterative} uses $ \mathcal{G} $ to compute the vector $ \MinInitCons_\mdp. $

\begin{theorem}
\label{thm:mininitconst-main}
Algorithm~\ref{algo:mininitcons_iterative} correctly computes the
vector $ \MinInitCons_\mdp $. Moreover, the repeat-loop terminates
after at most $|\states|$ iterations.
\end{theorem}

\begin{proof}
The repeat-loop performs the iteration of the operator $ \mathcal{G}
$. We show that the fixed point of the iteration equals $
\MinInitCons_\mdp$. Consider the iteration of $ \mathcal{F} $ and
$\mathcal{G}$ on $\infvec$ and $ \vec{x}_\reloads $, respectively. Let
$i$ be the number of steps (possibly infinite) after which the $
\mathcal{F} $-iteration reaches a fixed point and $ j $ the number of
steps after which the $ \mathcal{G} $-iteration reaches a fixed point.
We prove that $ i = j $. Indeed, for each step $ k $ we have that
\begin{equation}
\label{eq:fixpoint-minitcons}%
\veccomp{\mathcal{G}^{k+1}(\infvec)}{s}\neq\veccomp{\mathcal{G}^{k}(\infvec)}{s} \Leftrightarrow
\veccomp{\mathcal{F}^{k+1}(\vec {x}_\reloads)}{\tilde{s}}\neq\veccomp{\mathcal{F}^{k}(\vec{x}_\reloads)}{\tilde{s}}
\end{equation}
(by \Cref{lem:minitconst-reach-to-direct}). Hence, $ i \geq j $. For
the reverse inequality, assume that $i>j$. Then there is $t \in
\widetilde{\states}$ such that
$\veccomp{(\vec{x}_\reloads)}{t}\neq\veccomp{\mathcal{F}^{j}(\vec{x}_\reloads)}{t}$.%
From \eqref{eq:fixpoint-minitcons} and from the fact that the
$\mathcal{G}$-iteration already reached a fixed point we get that $ t
\in \states $. Then either $t\in \states \setminus \reloads$, but then
by \Cref{lem:minitconst-reach-to-direct} we have $
\veccomp{\mathcal{G}^{j+1}(\infvec)}{t} =
\veccomp{\mathcal{F}^{j+1}(\vec{x}_\reloads)}{t}  \neq
\veccomp{\mathcal{F}^{j}(\vec{x}_\reloads)}{t} =
\veccomp{\mathcal{G}^{j}(\infvec)}{t}$ a contradiction with
$\mathcal{G}$-iteration already being at a fixed point. Or  $ t \in
\states \cap \reloads $, but then $
\veccomp{\mathcal{F}^{j+1}(\vec{x}_\reloads)}{t}  =
\veccomp{\mathcal{F}^{j}(\vec{x}_\reloads)}{t} = 0$, again a
contradiction.

Hence, iterating $ \mathcal{G} $ also reaches a fixed point in at most
$ |\states| $-steps, by \Cref{thm:minreach-main}. Moreover, for each
$s\in \states$ we have $\veccomp{ \mathcal{G}^i(\infvec)}{s}  =
\veccomp{\mathcal{F}^i(\vec{x}_\reloads)}{\tilde{s}} =
\minreach{\widetilde{\mdp},\reloads}(\tilde{s}) =
\MinInitCons_\mdp(s)$, the first equality coming from
\Cref{lem:minitconst-reach-to-direct}, the second from
\Cref{thm:minreach-main} and from the fact that $ i=j $ and the last
one from \Cref{lem:mininitcons-minreach}.%
\qed%
\end{proof}

\subsection{Solving the Safety Problem}

We want to identify a set $\reloads'\subseteq R$ such that we can
reach $\reloads'$ in at least 1 step and with consumption at most $\Ca
= \Ca(\mdp)$, from each $r \in \reloads'$. This entails identifying the maximal
$\reloads'\subseteq R$ such that $\MinInitCons_{\mdp(\reloads')} \leq
\Ca$ for each $r\in \reloads'$. This can be done by initially setting $ \reloads'=\reloads $ and iteratively removing states that have $ \MinInitCons_{\mdp(\reloads')} >
\Ca $, from $ \reloads' $, as in Algorithm~\ref{algo:initload}.

\begin{algorithm}
\KwIn{CMDP $ \mdp $}
\KwOut{The vector $\safe_{\mdp}$}
$ \Ca \leftarrow \Ca(\mdp) $\;
$\relvar \leftarrow \reloads $; $\varToRemove \leftarrow \emptyset$\;
\Repeat{$\varToRemove = \emptyset$}{
  $ \relvar \leftarrow \relvar\smallsetminus \varToRemove $\;
  $ \vec{mic} \leftarrow \MinInitCons_{\mdp(\relvar)}$\label{algoline:mcon}\;
  $ \varToRemove \leftarrow
    \{ r \in \relvar \mid \veccomp{\vec{mic}}{r} > \Ca \}$\;
}
\ForEach{$ s\in \states $}{
\lIf{$ \veccomp{\vec{mic}}{s} > \Ca $}
  {$ \veccomp{\vec{out}}{s} = \infty $}
  \lElse{$ \veccomp{\vec{out}}{s}  = \veccomp{\vec{mic}}{s}$}
}
\Return{$\vec{out}$}
\caption{Computing the vector $\safe_{\mdp}$.}
\label{algo:initload}
\end{algorithm}

\begin{theorem}\label{thm:safety-main}
Algorithm~\ref{algo:initload} computes the vector $ \safe_{\mdp} $ in polynomial time.
\end{theorem}
\begin{proof}
The algorithm clearly terminates. Computing $  \MinInitCons_{\mdp(\relvar)} $ on line~\ref{algoline:mcon} takes a polynomial number of steps per call due to \Cref{thm:mininitconst-main} and since $ \mdp(\mathit{Rel}) $ has asymptotically the same size as $ \mdp $. Since the repeat loop performs at most $ |\reloads| $ iterations, the complexity follows.

As for correctness, we first prove that $ {\vec{out}}
\leq \safe_\mdp $. It suffices to prove for each $s\in\states$ that
upon termination, $\veccomp{\vec{mic}}{s} \leq \safe_{\mdp}(s)$
whenever the latter value is finite. Since $ \MinInitCons_{\mdp'}(s)
\leq \Safe_{\mdp'}(s) $ for each MDP $ \mdp' $ and each its state such
that $ \safe_{\mdp'}(s) <\infty $, it suffices to show that $
\safe_{\mdp(\relvar)} \leq \safe_{\mdp} $ is an invariant of the
algorithm (as a matter of fact, we prove that $ \safe_{\mdp(\relvar)}
= \safe_{\mdp} $). To this end, it suffices to show that at every
point of execution $ \safe_{\mdp}(t) = \infty $ for each $ t\in
\reloads\setminus \relvar $: indeed, if this holds, no strategy that
is safe for some state $s \neq t$ can play an action $a$ from $s$ such
that $t \in \Succ(s,a)$, so declaring such states non-reloading does
not influence the $ \safe_{\mdp} $-values. So denote by $ \relvar_i $
the contents of $ \relvar $ after the $ i $-th iteration. We prove, by
induction on $ i $, that $ \safe_{\mdp}(s) = \infty $ for all $ s \in
\reloads \setminus \relvar $. For $ i = 0 $ we have $ \reloads =
\relvar $, so the statement holds. For $ i>0 $, let $ s \in \reloads
\setminus \relvar_{i}$, and let $ \sigma $ be any strategy. If some
run from $\compatible{\sigma}{s}$ visits a state from $ \reloads
\setminus \relvar_{i-1} $, then $ \sigma $ is not $ \Ca $-safe, by
induction hypothesis. Now assume that all such runs only visit reload
states from $\relvar_{i-1}$. Then, since $
\MinInitCons_{\mdp(\relvar_{i-1})}(s)>\Ca $, there must be a run
$\run\in\compatible{\sigma}{s}$ with
$\reachcons{\relvar_{i-1}}^+(\run)>\Ca$. Assume that $ \run $ is $ \Ca $-safe in $ s $. Since we consider only decreasing CMDPs, $ \run $ must infinitely often visit a reload state (as it cannot get stuck in a zero cycle). Hence, there exists an index $ f>1 $ such that $ \rstate{f}\in\relvar_{i-1} $, and for this $ f $ we have
$\enlev{\Ca}{\run\pref{f}} = \bot$, a contradiction. So again, $\sigma$ is not safe in
$s$. Since there is no safe strategy from $s$, we have
$\safe_{\mdp}(s) = \infty$.

Finally, we need to prove that upon termination, $ \vec{out} \geq \safe_{\mdp} $. Informally, per the definition of $ \vec{out} $, from every state $ s $ we can ensure reaching a state of $ \mathit{Rel} $ by consuming at most $ \vec{out}(s) $ units of the resource. Once in $ \mathit{Rel} $, we can ensure that we can again return to $ \mathit{Rel} $ without consuming more than $ \Ca $ units of the resource. Hence, when starting with $ \vec{out}(s) $ units, we can surely prevent resource exhaustion. \qed
\end{proof}

\begin{definition}
\label{def:safeact} We call an action $ a $ \emph{safe} in a state $ s
$ if one of the following conditions holds:
\begin{itemize}
\item $ s \not \in \reloads $ and $ \cons(s,a) + \max_{t\in
\Succ(s,a)} {\safe_{\mdp}}(t) \leq \safe_{\mdp}(s) $; or%
\item $ s \in \reloads $ and $ \cons(s,a) + \max_{t\in \Succ(s,a)}
{\safe_{\mdp}}(t) \leq \Ca(\mdp) $.
\end{itemize}
Note that by the definition of $ \safe_{\mdp} ,$ for each state $ s $
with $ \safe_{\mdp}(s) < \infty$ there is always at least one action
safe in $ s $. For states $ s $ s.t. $ \safe_{\mdp}(s) = \infty $, we
stipulate all actions to be safe in $ s $.
\end{definition}

\begin{theorem}
\label{thm:safety-strat} Any strategy which always selects an action
that is safe in the current state is $ \safe_{\mdp}(s) $-safe in every
state $ s $. In particular, in each consumption MDP $ \mdp $ there is
a memoryless strategy $ \sigma $ that is $ \safe_{\mdp}(s) $-safe in
every state $ s $. Moreover, $ \sigma $ can be computed in polynomial time. \end{theorem}

\begin{proof}
The first part of the theorem follows directly from
\Cref{def:safeact}, \Cref{def:energy-level} (resource levels), and from
definition of $d$-safe runs. The second part is a corollary of
\Cref{thm:safety-main} and the fact that in each state, the safe
strategy from \Cref{def:safeact} can fix one such action in each
state and thus is memoryless. The complexity follows from Theorem~\ref{thm:safety-main}. \qed
\end{proof}
\section{Positive Reachability}%
\label{sec:posreach}
In this section, we focus on strategies that are safe and such that at
least one run they produce visits a given set $\target\subseteq S$ of
\emph{targets}. 
The main contribution of this section is
Algorithm~\ref{algo:posreach} used to compute such strategies as well as the vector $ \SafePosReach[\mdp,\target] $ of minimal initial resource levels for which such a strategy exist. As before, for the rest of this section we fix a CMDP $ \mdp $.

We define a function $\SPRval[\mdp]\colon
\states\times\actions\times\extNset^\states \to \extNset$ ($ \mathit{SPR} $ for safe positive reachability) s.t. for all
$s\in\states, a\in\actions$, and $\vec{x}\in\extNset^\states$
we have
\[
\SPRval_\mdp(s, a, \vec{x}) = \cons(s, a) + 
  \min_{t\in\Succ(s, a)} \Big\{ 
    \max \left\{\veccomp{\vec{x}}{t} , \safe_\mdp(t^\prime) \mid
    t'\in \Succ(s,a), t' \neq t 
\right\} \Big\}
\]
The $\max$ operator considers, for given $t$, the value $\veccomp{\vec{x}}{t}$
 and the values needed to survive from all possible outcomes of
$a$ other than $t$. Let $v = \SPRval_\mdp(s,a,\vec{x})$ and $t$ the
outcome selected by $\min$. Intuitively, $v$ is the minimal amount of
resource needed to reach $t$ with at least $\veccomp{\vec{x}}{t}$
resource units, or survive if the outcome of $a$ is different from $t$.

We now define a functional whose fixed point characterizes $ \SPRval[\mdp,\target] $. We first define a two-sided version of the truncation operator from the previous section: the operator $ \trunc[\mdp]{\cdot} $ such that
\[
\veccomp{\trunc[\mdp]{\vec{x}}}{s} =
\begin{cases}
  \infty &
    \text{if }\veccomp{\vec{x}}{s} > \Ca(\mdp)\\
  \vec{x}(s) &
    \text{if }
    \veccomp{\vec{x}}{s} \leq \Ca(\mdp)$ and $ s \not \in \reloads\\
  0 & \text{if }\veccomp{\vec{x}}{s} \leq \Ca(\mdp)$ and $ s \in \reloads\\
\end{cases}
\]
Using the functions $\SPRval$ and $ \trunc[\mdp]{\cdot} $, we now define an auxiliary operator $\mathcal{A}$ and the main operator $\mathcal{B}$ as follows.

\begin{align*}
\veccomp{\mathcal{A}_{\mdp}(\vec{r})}{s} &= 
  \begin{cases}
    \mathit{Safe}_{\mdp}(s) & \text{if } s \in \target \\
    { \min_{a\in A} \left(\SPRval[\mdp](s,a,\vec{r}) \right) } &
      \text{otherwise};
  \end{cases}\\
\mathcal{B}_{\mdp}(\vec{r}) &= \trunc{\mathcal{A}_\mdp(\vec{r})}_\mdp
\end{align*}

\noindent Let $\SafePosReach^i$ be the
vector such that for a state $s\in\states$ the number
$d=\SafePosReach^i(s)$ is the minimal number $ \ell $ such that there exists a
strategy that is $\ell$-safe in $s$ and produces at least one run that
visits $\target$ within first $i$ steps. Further, we denote by $
\vec{y}_\target $ a vector such that
\[
\veccomp{\vec{y}_\target}{s} = 
\begin{cases}
	\mathit{Safe}_{\mdp}(s) & \text{if } s \in \target \\
	\infty & \text{if } s \not\in \target 
\end{cases}
\]

The following lemma can proved by a rather straightforward but technical induction.

\begin{lemma}
\label{lem:posreach-iterate}
Consider the iteration of $\mathcal{B}_{\mdp}$ on the initial vector $
\vec{y}_\target $. Then for each $ i \geq 0 $ it holds that
$ \mathcal{B}_{\mdp}^i(\vec{y}_\target) = \SafePosReach[\mdp,\target]^i $.
\end{lemma}

The following lemma says that iterating $ \mathcal{B}_{\mdp} $ reaches a fixed point in a polynomial number of iterations. Intuitively, this is because when trying to reach $ \target $, it doesn't make sense to perform a cycle between two visits of a reload state (as this can only increase the resource consumption) and at the same time it doesn't make sense to visit the same reload state twice (since the resource is reloaded to the full capacity upon each visit). The proof is straightforward and is omitted in the interest of brevity.

\begin{lemma}%
\label{lem:posreach-bound}%
Let $ K = |\reloads| + (|\reloads|+1)\cdot(|\states|-|\reloads|+1)$.
Taking the same initial vector $ \vec{y}_T $ as in
\Cref{lem:posreach-iterate}, we have
$\mathcal{B}_{\mdp}^{K}(\vec{y}_T) = \SafePosReach[\mdp,\target]$.
\end{lemma}

The computation of $\SafePosReach[\mdp,\target]$ and of the associated
witness strategy is presented in Algorithm~\ref{algo:posreach}.

\begin{algorithm}[htb]
  \KwIn{CMDP $\mdp$ with states $S$, set of target states $\target\subseteq S$}
  \KwOut{The vector $\SafePosReach[\mdp,\target]$, coreresponding rule selector $\selector$}
  $\vec{r} \leftarrow \{\infty\}^S  $\; \label{aline:prinitb}
  \ForEach{$s\in \states$ s.t. $ \safe_{\mdp}(s)<\infty $}
  {
    $\selector(s)(\safe_{\mdp}(s)) \leftarrow 
    \text{arbitrary action safe in }s $
  }\label{aline:initmid}
  \lForEach{$t\in \target$}
  {
    $\veccomp{\vec{r}}{t} \leftarrow \veccomp{\safe_\mdp}{t}$
    \label{aline:prinite}
  }
  \Repeat{$\,\vec{r}_{\mathit{old}} = \vec{r}$}
  {\label{aline:prloopb}
    $ \vec{r}_\old \leftarrow \vec{r} $\;
    \ForEach{$ s \in \states\setminus \target $}
    {
      $ \veccomp{\vec{a}}{s} \leftarrow 
        \arg\min_{a\in \actions}
        \SPRval(s, a, {\vec{r}_\old})$\;\label{aline:prasel}
      $\veccomp{\vec{r}}{s} \leftarrow
        \min_{a\in \actions} 
        \SPRval(s, a, {\vec{r}_\old})  $\;
    }
    $  \vec{r} \leftarrow \trunc{\vec{r}}_{\mdp}$\; \label{aline:prtrunc}
    \ForEach{$ s\in \states\setminus \target $}
    {
      \If{$ \veccomp{\vec{r}}{s} < \veccomp{\vec{r}_{\old}}{s}$
      \label{aline:cond}}
      {
        $\selector(s)(\veccomp{\vec{r}}{s}) \leftarrow
        \veccomp{\vec{a}}{s} $\; \label{aline:prinsert}
      }
    }
  }\label{aline:prloope}
  \Return{$\vec{r},\selector$}
  \caption{Positive reachability of $\target$ in $\mdp$}
  \label{algo:posreach}
\end{algorithm}

\begin{theorem}%
\label{thm:posreach-values-main}%
The Algorithm~\ref{algo:posreach} always terminates after a polynomial number of steps, and upon
termination, $ \vec{r} = \SafePosReach[\mdp,\target]$.
\end{theorem}

\begin{proof}
The repeat loop on lines \ref{aline:prinitb}--\ref{aline:prinite}
initialize $\vec{r}$ to $\vec{y}_\target$. The repeat loop on lines
\ref{aline:prloopb}--\ref{aline:prloope} then iterates the operator
$\mathcal{B}$. By \Cref{lem:posreach-bound}, the iteration reaches a
fixed point in at most $ K $ steps, and this fixed point equals $
\SafePosReach[\mdp,\target]$. The complexity bound follows easily, since $ K $ is of polynomial magnitude.
\end{proof}

The most intricate part of our analysis is extracting a strategy that is $ \target $-positive $ \SafePosReach[\mdp,\target](s) $-safe in every state $ s $.

\begin{theorem}%
\label{thm:posreach-strat-main}%
Let $ \vec{v} = \SafePosReach[\mdp,\target] $. Upon termination of
Algorithm~\ref{algo:posreach}, the computed selector $ \selector $ has
the property that the finite counter strategy $ \selector^{\vec{v}} $
is, for each state $s\in\states$, $\target$-positive
$\veccomp{\vec{v}}{s}$-safe in $s$. That is, a polynomial-size finite
counter strategy for the positive reachability problem can be computed
in polynomial time.
\end{theorem}

\noindent The rest of this section is devoted to the proof
of~\Cref{thm:posreach-strat-main}. The complexity follows from \Cref{thm:posreach-values-main}. Indeed, since the algorithm has a polynomial complexity, also the size of $ \selector $ is polynomial. The correctness  proof is based on the following
invariant of the main repeat loop: the finite counter strategy $\pi = \selector^{\vec{r}} $
has these properties:
\begin{enumerate}[(a)]
\item We have that $\pi$ is $\safe_\mdp(s)$-safe in every state
$s\in\states$; in particular, we have for
$l=\min\{\vec{r}(s),\Ca(\mdp)\}$ that $\enlev{l}{\hist} \neq \bot$ for
every finite path $\hist$ produced by $\pi$ from $s$.%
\item For each state $s\in\states$ such that $\veccomp{\vec{r}}{s}
\leq \Ca(\mdp)$ there exists a $\pi$-compatible finite path
$\hist=s_1a_1s_2\ldots s_n$ such that $s_1=s$ and $s_n\in\target$ and
such that ``the resource level with initial load $\veccomp{\vec{r}}{s}$
never decreases below $\vec{r}$ along $ \hist$'', which means that for
each prefix $\hist\pref{i}$ of $\hist$ it holds
$\enlev{\veccomp{\vec{r}}{s}}{\hist\pref{i}}\geq
\veccomp{\vec{r}}{s_i}$.
\end{enumerate}

\noindent The theorem then follows from this invariant (parts (a) and
the first half of (b)) and from
\Cref{thm:posreach-values-main}. We start with the following
support invariant, which is easy to prove.

\begin{lemma}
\label{lem:pr-rinv} The inequality $ \vec{r}\geq \safe_{\mdp} $ is an
invariant of the main repeat-loop.
\end{lemma}

\paragraph{Proving part (a) of the main invariant.} We use the
following auxiliary lemma.
\begin{lemma}
\label{lem:posreach-strat-fc-safety} Assume that $ \selector $ is a
counter selector such that for all $s \in \states$ such that $
\safe(s)<\infty $:
\begin{compactenum}[(1.)]
\item $ \safe(s)\in \dom(\selector(s)) $.%
\item For all $ x \in \dom(\selector(s)) $, for $ a = \selector(s)(x)
$ and for all $ t\in \Succ(s,a) $ we have $\enlev{x}{sat} = d -
\cons(s,a)\geq \safe(t)$ where $d=x$ for $s\notin R$ and $d=\Ca(\mdp)$
otherwise.%
\end{compactenum}
Then for each vector $ \vec{y} \geq \safe$ the strategy $ \pi=\selector^{\vec{y}} $ is $ \safe (s)$-safe in every state $ s $.
\end{lemma}

\begin{proof}
Let $ s $ be a state such that $ \veccomp{\vec{y}}{s} < \infty$. It
suffices to prove that for every $\pi$-compatible finite path $\fpath$
started in $s$ it holds $\bot\neq
\enlev{\veccomp{\vec{y}}{s}}{\fpath}$. We actually prove a stronger
statement: $\bot\neq \enlev{\veccomp{\vec{y}}{s}}{\fpath} \geq
\safe(\last{\fpath})$. We proceed by induction on the length of $
\fpath $. If $ \len{\fpath}=0 $ we have $ \enlev{\vec{y}(s)}{\fpath} =
\vec{y}(s) \geq \safe_{\mdp}(s)\geq 0 $. Now let $ \fpath = \histpr
\histconc t_1 a t_2 $ for some shorter path $\histpr$ with
$\last{\histpr}=t_1$ and $a\in \actions$, $t_1, t_2 \in\states$. By
induction hypothesis, $ l = \enlev{\vec{y}(s)}{\beta}\geq
\safe_{\mdp}(t_1) $, from which it follows that $\safe_{\mdp}(t_1) <
\infty$. Due to (1.), it follows that there exists at least one $x \in
\dom(\selector(t_1))$ such that $x \leq l$. We select maximal $x$
satisfying the inequality so that $a = \selector(t_1)(x)$. We have
that $\enlev{\vec{y}(s)}{\hist} = \enlev{l}{t_1at_2}$ by definition
and from (2.) it follows that $ \bot \neq \enlev{x}{t_1at_2} \geq
\safe(t_2)\geq 0$. All together, as $l \geq x$ we have that
$\enlev{\vec{y}(s)}{\hist} \geq \enlev{x}{t_1at_2} \geq \safe(t_2)\geq
0$. \qed%
\end{proof}

Now we prove the part (a) of the main invariant. We show that
throughout the execution of Algorithm 3, $\selector $ satisfies the
assumptions of \Cref{lem:posreach-strat-fc-safety}. Property (1.) is
ensured by the initialization on line~\ref{aline:initmid}. The
property (2.) holds upon first entry to the main loop by the
definition of a safe action (\Cref{def:safeact}). Now assume that $
\selector(s)(\vec{r}(s)) $ is redefined on line~\ref{aline:prinsert},
and let $ a $ be the action $ \vec{a}(s) $.

We first handle the case when $ s\not\in \reloads $. Since $ a $ was
selected on line~\ref{aline:prasel}, from the definition of $ \SPRval$
we have that there is  $ t \in \Succ(s,a) $ such that after the loop
iteration,
\begin{equation}
\label{eq:prinv-a} \vec{r}(s) = \cons(s,a) + \max\{\vec{r}_{\old}(t),
\safe(t') \mid t\neq t'\in\Succ(s,a)\} \geq \cons(s,a)+\max_{t'
\in \Succ(s,a)}\safe_{\mdp}(t'),
\end{equation}%
the latter inequality following from \Cref{lem:pr-rinv}. Satisfaction
of property (2.) in $ s $ then follows immediately from the equation
\eqref{eq:prinv-a}.

If $ s  \in \reloads $, then \eqref{eq:prinv-a} holds before the
truncation on line \ref{aline:prtrunc}, at which point $
\vec{r}(s)<\Ca(\mdp) $. Hence, $ \Ca(\mdp) -\cons(s,a)\geq \max_{t \in
\Succ(s,a)}\safe_{\mdp}(t) $ as required by (2.).
From \Cref{lem:posreach-strat-fc-safety,lem:pr-rinv} it follows that $
\Sigma^{\vec{r}} $ is $ \safe_{\mdp}(s) $-safe in every state $ s $.
This finishes the proof of part (a) of the invariant.

\paragraph{Proving part (b) of the main invariant.} Clearly, (b) holds
right after initialization. Now assume that an iteration of the main
repeat loop was performed. Denote $ \pi_{\old} $ denote the strategy $
\selector^{\vec{r}_{\old}} $ and by $ \pi $ the strategy
$\selector^{\vec{r}}$. Let $s$ be any state such that $
\veccomp{\vec{r}}{s} \leq \Ca(\mdp) $. If $ \veccomp{\vec{r}}{s} =
\veccomp{\vec{r}_{\old}}{s} $, then we claim that (b) follows directly
from the induction hypothesis: indeed, by induction hypothesis we have
that there is an $ s $-initiated $ \pi_{\old} $-compatible path $
\fpath $ ending in a target state s.t. the $
\veccomp{\vec{r}_{\old}}{s} $-initiated resource level along $ \fpath
$ never drops  $ \vec{r}_{\old} $, i.e. for each prefix $ \beta $ of $
\fpath $ it holds $ \enlev{\veccomp{\vec{r}_{\old}}{s}}{\beta}\geq
\vec{r}_{\old}{(\last{\beta})} $. But then $ \beta $ is also $ \pi
$-compatible, since for each state $ q $, $ \Sigma(q) $ was only
redefined for values smaller than $ \veccomp{\vec{r}_{\old}}{q} $.

The case when $ \veccomp{\vec{r}}{s} < \veccomp{\vec{r}_{\old}}{s}$ is
treated similarly. As in the proof of part (a), denote by $ a $ the
action $ \veccomp{\vec{a}}{s} $ assigned on line \ref{aline:prinsert}.
There must be a state $ t \in \Succ(s,a) $ s.t. \eqref{eq:prinv-a}
holds before the truncation on line \ref{aline:prtrunc}. In
particular, for this $ t $ it holds $ \enlev{\veccomp{\vec{r}}{s}}{sat} \geq
\vec{r}_{\old}(t)$. By induction hypothesis, there is a $ t
$-initiated $ \pi_{\old} $-compatible path $ \histpr $ ending in $
\target $ satisfying the conditions in (b). We put $ \fpath = s a
t\histconc\histpr $. Clearly $\fpath$ is $ s $-initiated and reaches $
\target $. Moreover, it is $ \pi $-compatible. To see this, note that 
$\selector^{\vec{r}}(s)(\veccomp{\vec{r}}{s}) = a$; moreover, the
resource level after the first transition is $ e(t) =
\enlev{\veccomp{\vec{r}}{s}}{sat} \geq \vec{r}_{\old}(t) $, and due to
the assumed properties of $ \histpr $, the $ \vec{r}_{\old}(t)
$-initiated resource level (with initial load $e(t)$) never decreases
below $ \vec{r}_{\old} $ along $\histpr$. Since $ \selector $ was only
re-defined for values smaller than those given by the vector $
\vec{r}_{\old} $, $ \pi $ mimics $ \pi_{\old} $ along $ \histpr $.
Since $ \vec{r} \leq \vec{r}_{\old} $, we have that along $ \fpath $,
the $ \veccomp{\vec{r}}{s} $-initiated resource level never decreases
below $ \vec{r} $. This finishes the proof of part (b) of the
invariant and thus also the proof of \Cref{thm:posreach-strat-main}\qed
\section{Büchi}
\label{sec:buchi}

This section proofs \Cref{thm:intro-main} which is the main
theoretical result of the paper. The proof is broken down into the
following steps.
\begin{enumerate}[(1.)]
\item We identify a largest set $\reloads^\prime\subseteq\reloads$ of
reload states such that from each $r\in\reloads^\prime$ we can reach
$\reloads^\prime$ again (in at least one step) while consuming at most
$\Ca$ resource units and restricting ourselves only to strategies that (i) avoid
$\reloads\setminus\reloads^\prime$ and (ii) guarantee
positive reachability of $\target$ in $\mdp(\reloads^\prime)$.%
\item We show that $\SafeBuchi[\mdp,\target]=\SafePosReach[\mdp(\reloads^\prime),\target]
$ and that the corresponding strategy
(computed by Algorithm~\ref{algo:posreach}) is also $\target$-Büchi
$\veccomp{\SafeBuchi[\mdp,\target]}{s}$-safe for each $s\in\states$.%
\end{enumerate}

Algorithm~\ref{algo:buchi-rel} solves (1.) in a similar fashion as
Algorithm~\ref{algo:initload} handled safety. In each iteration, we declare all states from which positive reachability of $T$ and safety within $\mdp(\relvar)$ cannot be guaranteed as non-reloading. This is repeated until we reach a fixed point. The number of iterations is clearly
bounded by $|\reloads|$.

\begin{algorithm}[htb]
\caption{Almost-sure Büchi reachability of $\target$ in $\mdp$.}
\label{algo:buchi-rel}%
\KwIn{CMDP $ \mdp = (\states, \actions, \trans, \cons, \reloads,
\Ca)$, target states $\target\subseteq S$}%
\KwOut{The largest set $\relvar\subseteq\reloads$ such that
$\veccomp{\SafePosReach[\mdp(\relvar),\target]}{r} \leq \Ca$ for all
$r\in\relvar$.}%
$\relvar \leftarrow \reloads $; $\varToRemove \leftarrow \emptyset$\;%
\Repeat{$\varToRemove = \emptyset$}%
{%
  $\relvar \leftarrow \relvar\smallsetminus \varToRemove $\; $
  (\vec{reach},\selector) \leftarrow \SafePosReach[\mdp(\relvar),\,\target]$\; $
  \varToRemove \leftarrow \{ r \in \relvar \mid
  \veccomp{\vec{reach}}{r} > \Ca \}$\; }%
\Return{$\vec{reach},\selector$}
\end{algorithm}

\begin{theorem}\label{thm:buchi}%
Let $\mdp = (\states, \actions, \trans, \cons, \reloads, \Ca)$ be a
CMDP and $\target\subseteq\states$ be a target set. Moreover, let $\reloads^\prime$ be the contents of $ \relvar $
upon termination of Algorithm~\ref{algo:buchi-rel} for the input $\mdp$ and
$\target$. Finally let $\vec{r}$ and $\selector$ be the vector
 and the selector
returned by
Algorithm~\ref{algo:posreach} for the input $\mdp$ and
$\target$. Then for every state $ s $, the finite counter strategy
$\sigma=\selector^\vec{r}$ is $T$-Büchi $\veccomp{\vec{r}}{s}$-safe in
$s$ in both $\mdp(\reloads^\prime)$ and $\mdp$. Moreover, the vector
$\vec{r}$ is equal to $\SafeBuchi[\mdp, \target]$.
\end{theorem}

\begin{proof}
We first show that $\sigma$ is $T$-Büchi $\veccomp{\vec{r}}{s}$-safe
in $\mdp(\reloads^\prime)$ for all $s\in\states$ with
$\veccomp{\vec{r}}{s} \leq \Ca$. Clearly it is $ \veccomp{\vec{r}}{s} $-safe, so it remains to prove that $ T $ is visited infinitely often with probability 1.
We know that upon every visit of a state $ r\in \reloads' $, $ \sigma $ guarantees a future visit to $ T $ with positive probability. As a matter of fact, since $ \sigma  $ is a finite memory strategy, there is $ \delta>0 $ such that upon every visit of some $ r\in \reloads' $, the probability of a future visit to $ T $ is at least $ \delta $.
 As $\mdp(\reloads^\prime)$ is
decreasing, every $ s $-initiated $ \sigma $-compatible run must
visit the set $\reloads^\prime$ infinitely many times. 
Hence, with probability 1 we reach $ T $ at least once. The argument can then be repeated from the first point of visit to $ T $ to show that with probability 1 ve visit $ T $ at least twice, three times, etc. \emph{ad infinitum.} By the  monotonicity of probability, $ \probm{\sigma}{\mdp,s}(\BuchiRuns[T]) =1$.

It remains to show that $ \vec{r}\leq \SafeBuchi[\mdp,T] $. Assume that there is a state $s\in\states$ and a strategy
$\sigma^\prime$ such that $\sigma^\prime$ is $d$-safe in $s$ for some
$d<\veccomp{\vec{r}}{s} = \veccomp{\SafePosReach[\mdp(\reloads'),
\target]}{s}$. We show that this strategy is not $T$-Büchi $d$-safe in
$\mdp$. If all $\sigma^\prime$-compatible runs reach $\target$, then there
must be at least one history $\hist$ produced by $\sigma^\prime$ that
visits $r\in\reloads\setminus\reloads^\prime$ before reaching
$\target$ (otherwise $d\geq \veccomp{\vec{r}}{s}$). Then either (a)
$\veccomp{\SafePosReach[\mdp, \target]}{r} = \infty$, in which case any $ \sigma'$-compatible extension of $ \hist $ avoids $ T $; or (b) since $ \veccomp{\SafePosReach[\mdp(\reloads'), \target]}{r} >\Ca $, there must be an extension of $ \hist $ that visits, between the visit of
$r$ and $\target$, another $r^\prime \in
\reloads\setminus\reloads^\prime$ such that $r^\prime \neq r$. We can then repeat the argument, eventually reaching the case (a) or running out of the resource, a contradiction with $\sigma^\prime$ being $d$-safe. \qed%
\end{proof}

We can finally proceed to prove \Cref{thm:intro-main}. 

\begin{proof}[of \Cref{thm:intro-main}]
The theorem follows immediately from \Cref{thm:buchi} since we can
(1.) compute $\SafeBuchi[\mdp, \target]$ and the
corresponding strategy $\sigma_\target$ in polynomial time (see
\Cref{thm:posreach-strat-main} and Algorithm~\ref{algo:buchi-rel}),
(2.) we can easily check whether $d\geq\veccomp{\SafeBuchi[\mdp,
\target]}{s}$, if yes, than $\sigma_\target$ is the desired strategy
$\sigma$ and (3.) represent $\sigma_T$ in polynomial space as it is a
finite counter strategy represented by a polynomial-size counter selector.
\qed
\end{proof}

\section{Implementation and Case Studies}
\label{sec:exp}

We have implemented the presented algorithms in Python in a tool
called \emph{FiMDP (Fuel in MDP)} available at
\url{https://github.com/xblahoud/FiMDP}. The docker artifact is
available at \url{https://hub.docker.com/r/xblahoud/fimdp} and can be run without installation via the
Binder project~\cite{binder.18.scipy}. 
We investigate the practical behavior of our algorithms using two case
studies: (1) An autonomous electric vehicle (AEV) routing problem in
the streets of Manhattan modeled using realistic traffic and electric
car energy consumption data, and (2) a multi-agent grid world model
inspired by the Mars Helicopter Scout \cite{balaram2018mars} to be
deployed from the planned Mars 2020 rover. The first scenario
demonstrates the utility of our algorithm for solving real-world
problems~\cite{Zhaetal19}, while the second scenario reaches the
algorithm's scalability limits.

The consumption-Büchi objective can be also solved by a naive approach
that encodes the energy constraints in the
state space of the MDP, and solves it using techniques for standard
MDPs~\cite{Alfaro:thesis}. States of such an MDP are tuples $(s, e)$ where $s$ is a state
of the input CMDP and $e$ is the current level of energy. Naturally,
all actions that would lead to states with $e < 0$ lead to a special
sink state. The standard techniques rely on decomposition of the MDP
into maximal end-components (MEC). We implemented the explicit
encoding of CMDP into MDP, and the MEC-decomposition algorithm.

All computations presented in the following were performed on a PC
with Intel Core i7-8700 3.20GHz 12 core processor and a RAM of
16 GB running Ubuntu 18.04 LTS. All running times are means from
at least 5 runs and the standard deviation was always below
5\% among these runs.

\subsection{Electric Vehicle Routing}

\begin{wrapfigure}{R}{.45\textwidth}
\vspace{-1cm}
\centering
\includegraphics[width=.45\textwidth,trim={.5cm 1cm .5cm 1cm}, clip]%
{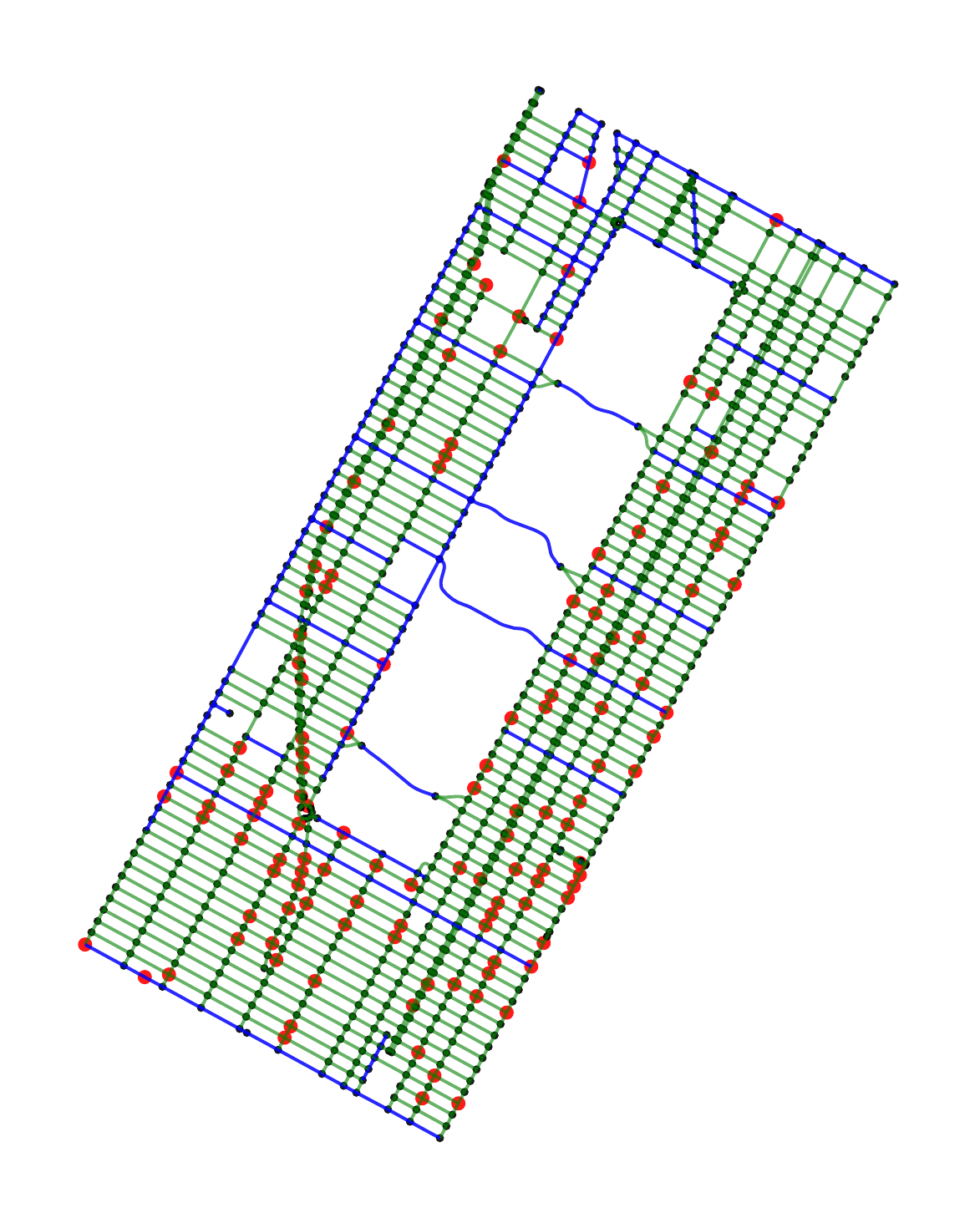}
\caption{Street network in the considered area. Charging stations are red, one way roads green, and two-way roads blue.} 
\label{fig:manhattan}
\end{wrapfigure}

We consider the area in the middle of Manhattan, from 42nd to 116th
Street, see Fig. \ref{fig:manhattan}. Street intersections and directions of feasible movement form the state and action spaces of the MDP. Intersections in the proximity of real-world fast charging stations~\cite{evdatabase} represent the set of reload states.

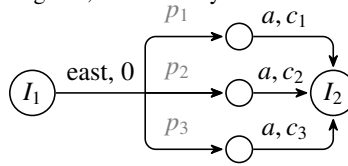
\begin{wrapfigure}[8]{r}{.45\textwidth}
\vspace{-1.8cm}
\centering
  \begin{tikzpicture}[automaton]
  \tikzstyle{action} = [circle, minimum size=0pt, draw=none]
  \tikzstyle{mstate} = [state, minimum size=10pt, inner sep=0pt, text width=0]
  \tikzstyle{consumption} = [black]
    \node[state] (m) at (2,1) {$I_1$};
    \node[state] (e) at (6,1) {$I_2$};
    
    \node[action] (ae)at (3.5,1) {};
    \path[above, shorten >=0pt] (m) edge node {east, $0$}(ae.center);
    \node[mstate](e0) at (4.75,.25) {};
    \node[mstate](e1) at (4.75,1) {};
    \node[mstate](e2) at (4.75,1.75) {};

      \path[->]
      (ae.center) edge[to path={|- 
          node[above, inner sep=-3pt, pos=.7, prob] {$p_3$} 
          (\tikztotarget)}, rounded corners]
      (e0)
      (e0) edge[to path={-| node[pos=.2,above]{$a, c_3$}(\tikztotarget)}, rounded corners] (e)
      
      (ae.center) edge node[above, inner sep=-3pt, pos=.4, prob] {$p_2$} (e1)
      (ae.center) edge[to path={-- (\tikztotarget)}, rounded corners] (e1)
      (e1) edge[] node[above] {$a, c_2$}(e)
      
      (ae.center) edge[to path={|- 
          node[above, inner sep=-3pt, pos=.7, prob] {$p_1$} 
          (\tikztotarget)}, rounded corners]
      (e2)
      (e2) edge[to path={-| node[pos=.2,above]{$a, c_1$} (\tikztotarget)}, rounded corners] (e);
  \end{tikzpicture}
\caption{Transition from intersection $I_1$ to $I_2$ with stochastic consumption. The small circles are dummy states.}
\label{fig:discretize}
\end{wrapfigure}

After the AEV picks a direction, it reaches the next intersection in
that direction deterministically with a stochastic energy consumption. We base our model of consumption on distributions
of vehicle travel times from the area \cite{Uber:2019} and conversion
of velocity and travel times to energy consumption \cite{Tesla:2008}.
We discretize the consumption distribution into three possible values
($c_1, c_2, c_3$) reached with corresponding probabilities ($p_1, p_2,
p_3$). The transition from one intersection ($I_1$) to another ($I_2$)
is then modelled using three dummy states as explained in
Fig.~\ref{fig:discretize}.

In this fashion, we model the street network of Manhattan as a CMDP with with $7378$ states and $8473$ actions. For a
fixed set of 100 randomly selected target states,
Fig.~\ref{fig:nyc-plots} shows influence of requested capacity on
running times for \textbf{(a)} strategy for Büchi objective using CMDP
(our approach), and \textbf{(b)} MEC-decomposition for the
corresponding explicit MDP. 
We can see from the plots that our algorithm runs
reasonably fast for all capacities (it stabilizes for $\Ca > 95$), it
is not the case for the explicit approach. The running times for
MEC-decomposition is dependent on the numbers of states and actions in
the explicit MDP, which keep growing. The number of states of the
explicit MDP for capacity 95 is 527475, while it is still only 7378 in
the original CMDP. Also note that actually solving the Büchi objective
in the explicit MDP requires computing almost-sure reachability of
MECs with some target states. Therefore, we can expect that even for
small capacities our approach would outperform the explicit one
(Fig.~\ref{fig:nyc-plots} (c)).

\begin{figure}[hbt]
\centering \scalebox{.7}{
%
\definecolor{mycolor1}{rgb}{0.00000,0.44700,0.74100}%
\definecolor{mycolor2}{rgb}{0.74100,0.00000,0.44700}%
\begin{tikzpicture}
\def\plotsize{1.1in}

\begin{axis}[%
width=\plotsize,
height=\plotsize,
at={(0,0)},
scale only axis,
restrict x to domain=10:195,
xmin=0,
xmax=200,
xlabel style={font=\color{white!15!black}},
xlabel={capacity},
ymin=0,
ylabel style={font=\color{white!15!black}},
ylabel={comp time (sec)},
axis background/.style={fill=white},
title style={font=\bfseries},
title={(a) CMDP},
legend style={legend cell align=left, align=left, draw=white!15!black}
]
\addplot[only marks, mark=*, mark options={}, mark size=1.5000pt, color=mycolor1, fill=mycolor1] table[col sep=comma,x index=1, y index=2] {data/comptime_diffcaps.csv};

\end{axis}

\begin{axis}[%
width=\plotsize,
height=\plotsize,
at={(4cm,0)},
scale only axis,
restrict x to domain=10:210,
xmin=0,
xmax=200,
xlabel style={font=\color{white!15!black}},
xlabel={capacity},
ymin=0,
axis background/.style={fill=white},
title style={font=\bfseries},
title={(b) explicit},
legend style={legend cell align=left, align=left, draw=white!15!black}
]
\addplot[only marks, mark=x, mark options={}, mark size=1.5000pt, color=mycolor2, fill=mycolor2]table[col sep=comma,x index=1, y index=3] {data/comptime_explicit.csv};

\end{axis}

\begin{axis}[%
width=\plotsize,
height=\plotsize,
at={(8cm,0)},
scale only axis,
restrict x to domain=10:210,
xmin=0,
xmax=60,
xlabel style={font=\color{white!15!black}},
xlabel={capacity},
ymin=0,
axis background/.style={fill=white},
title style={font=\bfseries},
title={(c) combined},
legend style={legend cell align=left, align=left, draw=white!15!black, font=\scriptsize},
legend columns=1, legend pos=north west,
legend entries={CMDP, MEC-decomp.}
]
\addplot[only marks, mark=*, mark options={}, mark size=1.5000pt, color=mycolor1, fill=mycolor1] table[col sep=comma,x index=1, y index=2] {data/comptime_diffcaps.csv};

\addplot[only marks, mark=x, mark options={}, mark size=1.5000pt, color=mycolor2, fill=mycolor2]table[col sep=comma,x index=1, y index=3] {data/comptime_explicit.csv};

\end{axis}

%

\pgfplotsset{scaled y ticks=false}
\end{tikzpicture}%
}%
  \caption{Mean computation times for a fixed target set of size 100
  and varying capacity: \textbf{(a) CMDP} -- computating Büchi
  objective via CMDP, \textbf{(b) explicit} -- computating MEC
  decomposition of the explicit MDP,  \textbf{(c) combined} --
  \textbf{(a)} and \textbf{(b)} combined for small capacity values.}
  \label{fig:nyc-plots}
\end{figure}
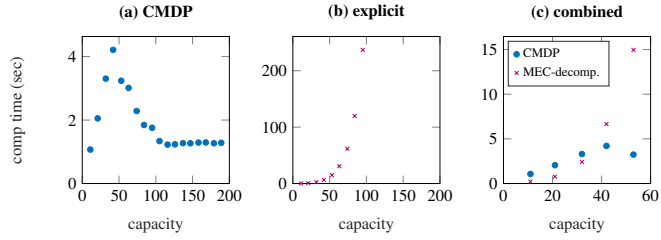

%

\subsection{Multi-agent Grid World}

We use multi-agent grid world to generate CMDP with huge number of
states to reach the scalability limits of the proposed algorithms. We
model the rover and the helicopter of the Mars 2020 mission with the
following realistic considerations: the rover enjoys infinite energy
while the helicopter is restricted by batteries recharged at the
rover. These two vehicle jointly operate on a mission where the
helicopter reaches areas inaccessible to the rover. The outcomes of
the helicopter's actions are deterministic while those of the rover
--- influenced by terrain dynamics --- are stochastic. For a grid
world of size $n$, this system can be naturally modeled as a CMDP with
$n^4$ states. Fig. \ref{fig:mars-plots} shows the running times of the Büchi objective for growing grid sizes and capacities in CMDP. It also shows the running time for the MEC decomposition of the corresponding explicit MDP when the capacity is 10. We observe that the increase in the computational time of CMDP follows the growth in the number of states roughly linearly, and our implementation deals with an MDP with $1.6\times 10^5$ states in no more than seven minutes.
\begin{figure}[hbt]
	\centering \scalebox{.7}{
%
\definecolor{mycolor1}{rgb}{0.00000,0.44700,0.74100}%
\definecolor{mycolor2}{rgb}{0.74100,0.00000,0.44700}%
\definecolor{mycolor3}{rgb}{0.3000,0.64100,0.24700}%
\definecolor{mycolor4}{rgb}{1.0000,0.31100,0.0000}%
\begin{tikzpicture}
\def\plotsize{1.1in}

\begin{axis}[%
width=\plotsize,
height=\plotsize,
at={(0,0)},
scale only axis,
restrict x to domain=0:25,
xmin=0,
xmax=22,
xlabel style={font=\color{white!15!black}},
xlabel={grid size},
ymin=0,
ylabel style={font=\color{white!15!black}},
ylabel={comp time (sec)},
axis background/.style={fill=white},
title style={font=\bfseries},
title={(a) CMDP},
legend style={legend cell align=left, align=left, draw=white!15!black, font=\scriptsize},
legend columns=1, legend pos=north west,
legend entries={cap = 10,cap = 50,cap = 100}
]
\addplot[only marks, mark=*, mark options={}, mark size=1.2000pt, color=mycolor1, fill=mycolor1] table[col sep=comma,x index=1, y index=3] {data/mars_comptime_diffgs_cap10.csv};
\addplot[only marks, mark=triangle*, mark options={}, mark size=1.5000pt, color=mycolor3, fill=mycolor3]table[col sep=comma,x index=1, y index=3] {data/mars_comptime_diffgs_cap50.csv};
\addplot[only marks, mark=diamond*, mark options={}, mark size=1.5000pt, color=mycolor4, fill=mycolor2]table[col sep=comma,x index=1, y index=3] {data/mars_comptime_diffgs_cap100.csv};

\end{axis}

\begin{axis}[%
width=\plotsize,
height=\plotsize,
at={(4cm,0)},
scale only axis,
restrict x to domain=0:25,
xmin=0,
xmax=22,
xlabel style={font=\color{white!15!black}},
xlabel={grid size},
ymin=0,
axis background/.style={fill=white},
title style={font=\bfseries},
title={(b) explicit},
legend style={legend cell align=left, align=left, draw=white!15!black}
]
\addplot[only marks, mark=x, mark options={}, mark size=1.5000pt, color=mycolor2, fill=mycolor2]table[col sep=comma,x index=1, y index=3] {data/mars_comptime_explicit.csv};

\end{axis}

\begin{axis}[%
width=\plotsize,
height=\plotsize,
at={(8cm,0)},
scale only axis,
restrict x to domain=0:25,
xmin=0,
xmax=22,
xlabel style={font=\color{white!15!black}},
xlabel={grid size},
ymin=0,
axis background/.style={fill=white},
title style={font=\bfseries},
title={(c) combined},
legend style={legend cell align=left, align=left, draw=white!15!black, font=\scriptsize},
legend columns=1, legend pos=north west,
legend entries={CMDP, MEC-decomp.}
]
\addplot[only marks, mark=*, mark options={}, mark size=1.2000pt, color=mycolor1, fill=mycolor1] table[col sep=comma,x index=1, y index=3] {data/mars_comptime_diffgs_cap10.csv};

\addplot[only marks, mark=x, mark options={}, mark size=1.5000pt, color=mycolor2, fill=mycolor2]table[col sep=comma,x index=1, y index=3] {data/mars_comptime_explicit.csv};

\end{axis}

%

\pgfplotsset{scaled y ticks=false}
\end{tikzpicture}%
}%
	\caption{Mean computation times for varying grid sizes and of size capacities: \textbf{(a) CMDP} -- computating Büchi
		objective via CMDP, \textbf{(b) explicit} -- computating MEC
		decomposition of the explicit MDP for a capacity of 10,  \textbf{(c) combined} -- combined computation time for a capacity of 10.}
	\label{fig:mars-plots}
\end{figure}
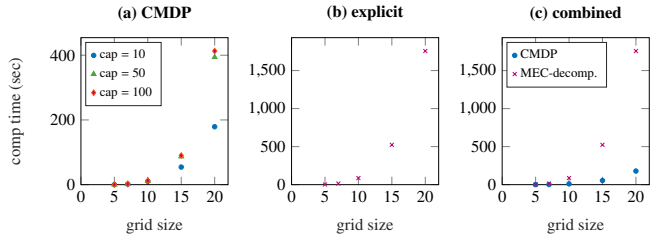
\vspace{-0.3cm}
\section{Conclusion \& Future Work}
We presented a first study of consumption Markov decision processes (CMDPs) with qualitative $ \omega $-regular objectives. We developed and implemented a polynomial-time algorithm for CMDPs with an objective of probability-1 satisfaction of a given B\"uchi condition. Possible directions for the future work are extensions to quantitative analysis (e.g. minimizing the expected resource consumption), stochastic games, or partially observable setting.

\noindent\textbf{Acknowledgements:} We acknowledge the kind help of Vojtěch Forejt, David Klaška, and Martin Kučera in the discussions leading to this paper.
\bibliographystyle{abbrv}
\bibliography{petr,pranay}
\newpage
\appendix
\begin{center}
\Large Technical Appendix
\end{center}
\section{Proofs}
\subsection{Proof of Theorem~\ref{thm:minreach-main}}
\label{app:minreachability}

\begin{lemma}
 \label{lem:reachopt} There exists a memory-less optimal strategy for
 the objective $\minreach\target$.
 \end{lemma}
 \begin{proof}
 It is clear that in every state $s$, the player must play a
 \emph{good} action, i.e. an action $a$ such that $\cons(s,a) +
 \max_{s' \in \Succ(s,a) }\minreach\target(s') \leq\minreach\target(s)
 $. If there are multiple good actions in some state $s$ we proceed as
 follows.
 
 We first assign a \emph{ranking} to states. All states in $\target$ have
 rank $0$. Now assume that we have assigned ranks $\leq i$ and we want
 to assign rank $i+1$ to some of the yet unranked states. We say that
 an action $a$ is \emph{progressing} in an unranked state $s$ if all
 $s'\in \Succ(s,a)$ have a rank (smaller than $i+1$). If all good
 actions in the yet unranked states are non-progressing, than all the
 unranked states $s$ have $\minreach\target(s)=\infty$, since by
 playing any good action, the player cannot force reaching a ranked
 state and thus also a target state; hence, in this case we assign all
 the unranked states the rank $\infty$ and finish the construction. 
 Otherwise, we assign rank $i+1$ to all unranked states that have a
 good progressing action and continue with the construction. It is easy
 to see that a state $s$ is assigned a finite rank if and only if
 $\minreach\target(s)<\infty$.
 
 We now fix a memory-less strategy $\sigma$ such that in a state of
 finite rank, $\sigma$ chooses a good progressing action; for states of
 an infinite rank, $\sigma$ chooses an arbitrary (but fixed) action.
 Since $\sigma$ only uses good actions, a straightforward induction
 shows that for each state $s$ and each $s$-initiated run that reaches
 $\target$ actually reaches $\target$ with consumption at most
 $\minreach\target(s)$. So we need to show that $\sigma$ does not admit
 runs initiated in a state of finite rank that never reach $\target$. But all
 $\sigma$-compatible runs initiated in a state of finite rank decrease
 the rank in every step, since $\sigma$ only plays progressing actions.
 The result follows.
 \qed
 \end{proof}

Given a target set $\target$, a number $i \in \Nset$, and a run
$\run$, we define $\reachcons{\target}^i(\run)$ as
$\reachcons{\target}(\run)$ with the additional restriction that $f
\leq i$. Intuitively, $\reachcons{\target}^i(\run)=\infty$ if $\run$
does not visit $\target$ withing $i-1$ steps. We then put
\[\minreach{\target}^i(s) = \inf_{\sigma}\sup_{\run \in
\compatible{\sigma}{s}} \reachcons{\target}^i(\run).\]

\begin{lemma}
\label{lem:stepreach-cost} For every $i\geq 0$ it holds that
$\minreach{\target}^i(s) = \veccomp{\mathcal{F}^i(\vec{x}_T)}{s} $.
\end{lemma}
\begin{proof}
By an induction on $i$. The base case is simple. Now assume that the equality holds for some $i\geq 0$. For $i+1$ we get:
\begin{align*}
\minreach{\target}^{i+1}(s) &= 
\inf_{\sigma}\sup_{\run \in \compatible{\sigma}{s}} 
  \reachcons{\target}^{i+1}(\varrho) \\
&= \min_{a\in\actions} \Big( \cons(s,a) + \max_{t \in \Succ(s,a)}  \inf_{\sigma}\sup_{\run \in \compatible{\sigma}{t}} \reachcons{\target}^{i}(\varrho)\Big)\\
&\stackrel{\mathclap{\tiny\mbox{(I.H.)}}}{=} \min_{a\in\actions} \Big( \cons(s,a) + \max_{t \in \Succ(s,a)} \veccomp{\mathcal{F}^i(\vec{x}_\target)}{t} \Big) = \veccomp{\mathcal{F}^{i+1}(\vec{x}_\target)}{s}
\end{align*}
\qed
\end{proof}

\begin{reftheorem}{thm:minreach-main}
Denote by $n$ the length of the longest
simple path in $\mdp$. Let $\vec{x}_T$ be a vector such that
$\veccomp{\vec{x}_T}{s}=0$ if $s\in \target$ and
$\veccomp{\vec{x}_T}{s}=\infty$ otherwise. Then iterating
$\mathcal{F}$ on $\vec{x}_T$ yields a fixpoint in at most $n$ steps
and this fixpoint equals $\minreach{\target}$.
\end{reftheorem}

\begin{proof}
It is easy to see that for each $s$ and each $i\geq 0$ it holds that
$\minreach\target^i(s) \geq \minreach\target(s)$. Furthermore, an easy
computation shows that $\minreach\target(s)$ is a fixed point of
$\mathcal{F}$.  Hence, it suffices to show
that $\minreach\target^{|\states|}(s) = \minreach\target(s)$. Assume,
for the sake of contradiction, that we have some $s\in \states$ such
that $\minreach\target^{|\states|} (s) > \minreach\target(s)$. Fix
$\sigma$ to be the memory-less optimal strategy from
\Cref{lem:reachopt}. By the definition of $
\minreach\target^{|\states|} (s) $ we have a run $\run \in
\compatible{\sigma}{s}$ such that $
\reachcons\target^{|\states|}(\run) > \minreach\target(s) $. This is
only possible if $\reachcons\target^{|\states|}(\run) = \infty$,
otherwise this would contradict the optimality of $\sigma$ (note that
if $\reachcons\target^i(\run)<\infty$, then $
\reachcons\target^j(\run) = \reachcons\target^i(\run)$ for all $ j
\geq i $). Hence, there is a run compatible with $\sigma$ whose prefix
of length $|\states|$ does not contain a target state. But this prefix
must contain a cycle, and since $\sigma$ is memory-less, we can
iterate this cycle forever. Hence, the optimal strategy $\sigma$
admits a run from $s$ that never reaches $T$, a contradiction with
$\minreach\target(s) < \minreach\target^{|\states|}(s) \leq \infty$.%
\qed
\end{proof}

\subsection{Proof of~\Cref{lem:minitconst-reach-to-direct}}

By induction on $ i $. The base case is clear. Now assume that the
equality holds for some $ i \geq 0 $. Note that for any
$(s,a)\in\states \times \actions$ we have $\Succ(\tilde{s},a) =
\Succ(s,a)\subseteq \states$. Also, for all $s' \in \states \setminus
\reloads$ we have $\veccomp{\strunc{\mathcal{G}^i(\infvec)}{}}{s'} =
\veccomp{{\mathcal{G}^i(\infvec)}}{s'} = \veccomp{\mathcal{F}^i(\vec{x}_\reloads)}{s'}$
(by induction hypothesis), while for  $s' \in \states \cap \reloads$
we have $ \veccomp{\strunc{\mathcal{G}^i(\infvec)}{}}{s'} = 0 =
\veccomp{\mathcal{F}^i(\vec{x}_\reloads)}{s'}$ (by the definition of $
\mathcal{F} $). Hence, $\max_{s'\in
  \Succ(s,a)}\veccomp{\strunc{\mathcal{G}^i(\infvec)}{}}{s'} = \max_{s'\in
  \Succ(\tilde{s},a)}{\veccomp{\mathcal{F}^i(\vec{x}_\reloads)}{s'}}$. 
Thus,
\begin{align*}
\veccomp{\mathcal{G}^{i+1}(\infvec)}{s} &=  \min_{a\in \actions} \left(\cons(s,a) + \max_{s'\in \Succ(s,a)} \veccomp{\strunc{\mathcal{G}^i(\infvec) }{}}{s'}\right) \\
& = \min_{a\in \actions} \left(\cons(s,a) + \max_{s'\in \Succ(\tilde{s},a)}{\veccomp{\mathcal{F}^i(\vec{x}_\reloads)}{s'}}\right)\\
& = \veccomp{\mathcal{F}^{i+1}(\vec{x}_\reloads)}{\tilde{s}}.
\end{align*}
(The last equation following from the definition of $ \mathcal{F} $ and from the fact that $\tilde{s}$ is never a reload state.)

Moreover, from \Cref{lem:stepreach-cost} it follows that $
\veccomp{\mathcal{F}^{i+1}(\vec{x}_\reloads)}{s} =
\minreach{\reloads}^{i+1}(s) $. But if $s$ is not a reload state, then
$ \minreach{\reloads}^{i+1}(s)  = \minreach\reloads^{i+1}(\tilde{s})$,
so for $ s\in \states\setminus\reloads $ we have $
\veccomp{\mathcal{G}^{i+1}(\infvec)}{s} =
\veccomp{\mathcal{F}^{i+1}(\vec{x}_\reloads)}{\tilde{s}} =
\veccomp{\mathcal{F}^{i+1}(\vec{x}_\reloads)}{s} $, which proves the
second part. \qed%

\subsection{Completion of the proof of~\Cref{thm:safety-main}}
\label{app:safetymain}

Now we prove that upon termination, $ \vec{out} \geq \safe_{\mdp} $.
For every $ t\in \states $ there exists, by the definition of
$\MinInitCons$, a strategy $\sigma_t$ s.t. $\reachcons{\mdp(\relvar),
\relvar}^+(\sigma_t, t) = \veccomp{\vec{mic}}{t} =
\veccomp{\vec{out}}{t}$ and thus $\reachcons{\mdp(\relvar),
\relvar}^+(\run)$ is bounded by $\veccomp{\vec{out}}{t}$ for each
$\run\in\compatible{\sigma_t}{t}$. We construct a new strategy $\pi$
that, starting in some state $t$ initially mimics $\sigma_t$ until the
next visit of a state $r_1\in \relvar$. Once this happens, the
strategy $\pi$ begins to mimic $\sigma_{r_1}$ until the next visit of
some $r_2 \in \relvar$, when $\pi$ begins to mimic $\sigma_{r_2}$, and
so on \emph{ad infinitum.}

Fix any state $s$ such that upon termination, $\veccomp{\vec{mic}}{s}
\leq \Ca$ (for other states, the inequality $\veccomp{\vec{out}}{s}
\geq \safe_\mdp(s)$ clearly holds). We prove that every run
$\run\in\compatible[\mdp(\relvar)]{\pi}{s}$ is actually $
\veccomp{\vec{mic}}{s} $-safe from $ s $ in the original MDP $ \mdp $.
In fact, it is sufficient to show this in $ \mdp(\relvar) $, since
each its reload state is also a reload state of $\mdp$.

Let $i_1 < i_2 < \ldots$ be all indices $i$ such that
$\rstate{i}\in\relvar$. Since $\pi$ mimics $\sigma_s$ until
$\rstate{i_1}$, we have that
$\enlev{\veccomp{\vec{mic}}{s}}{\run\pref{j}} = \veccomp{\vec{mic}}{s}
- \pathcons(\run\pref{j})$ where $ \pathcons(\run\pref{j}) $ is by
definition of $\sigma_s$ bounded by $\veccomp{\vec{mic}}{s}$, for all
$j \leq i_1$. Now let $m \geq 1$, we set $k=i_m$ and $l=i_{m+1}$. As
$\pi$ mimics $\sigma_{\rstate{k}}$ for between $\rstate{k}$ and
$\rstate{l}$ we have for $k < j \leq l$ that
$\enlev{\veccomp{\vec{mic}}{s}}{\run\pref{j}} =
\enlev{\Ca}{\run\infix{k}{j}} = \Ca - \pathcons(\run\infix{k}{j})$,
where $\pathcons(\run\infix{k}{j})$ is bounded by
$\reachcons{\relvar}^+(\sigma_{\rstate{k}}, \rstate{k}) =
\veccomp{\vec{mic}}{\rstate{k}} < \Ca$. Therefore,
$\enlev{\veccomp{\vec{mic}}{s}}{\run\pref{j}} \geq 0$ for all $j$ and
$\run$ is $\veccomp{\vec{mic}}{s}$-safe. This finishes the proof.
\qed

\subsection{Proof of \Cref{lem:posreach-iterate}}
\label{app:posreach-iterate}
By induction on $ i $. The base case is clear. Now assume that the
statement holds for some $ i \geq 0 $. Fix any $s$. Denote by $b =
\veccomp{\mathcal{B}^{i+1}(\vec{y}_\target)}{s}$ and $
d=\SafePosReach^{i+1}(s) $. We show that $b = d$. The equality
holds whenever $s\in \target$, so in the remainder of the proof we assume
that $s\not \in \target$.

We first prove that $b\geq d$. If $b=\infty$, this is clearly true.
Otherwise, let $a_{\min}$ be the action minimizing
$\SPRval(s,a_{\min},\mathcal{B}^{i}(\vec{y}_T)) $ (which equals $b$
if $ s\not \in \reloads $) and let $t_{\min} \in \Succ(s,a_{\min})$ be
the successor used to achieve this value. By induction hypothesis,
there exists a strategy $\sigma_1$ that is
$\veccomp{\mathcal{B}^{i}(\vec{y}_T)}{t_{\min}}$-safe in $t_{\min}$
and $\probm{\sigma_1}{t_{\min}}(\ReachRuns^i)>0$, and there also
exists a strategy $\sigma_2$ that is $\Safe(t)$-safe from $t$ for all
$t\in \Succ(s,a_{\min}), t\neq t_{\min}$.

Consider now a strategy $\pi $ which, starting in $s$, plays
$a_{\min}$. If the outcome of $a_{\min}$ is $t_{\min}$, $\pi$ starts
to mimic $\sigma_1$, otherwise it starts to mimic $\sigma_2$. We claim
that $\pi$ is $b$-safe in $s$ and that
$\probm{\pi}{s}(\ReachRuns^{i+1})>0$ by showing the following two points.
\begin{compactenum}
\item There is at least one run $\run_\target\in\compatible{\pi}{s}$ that reaches $\target$ in $\leq i+1$ steps.%
\item All runs in $\compatible{\pi}{s}$ are $b$-safe.
\end{compactenum}
\noindent We construct $\run_\target$ easily as
$\run_\target=\hist\histconc\run^\prime$ where $\hist =
sa_{\min}t_{\min}$ and $\run^\prime$ is the witness that
$\probm{\sigma_1}{t_{\min}}(\ReachRuns^i)>0$. Now let
$\run\in\compatible{\pi}{s}$ be a run produced by $\pi$ from $s$. Then
it has to be of the form $\run=sa_{\min}t\histconc\run_\sigma$ where
$\run_\sigma\in\compatible{\sigma_1}{t}$ if $t=t_{\min}$ or
$\run_\sigma\in\compatible{\sigma_2}{t}$ it $t\neq t_{\min}$; in both
cases $\veccomp{\Safe}{t}$-safe in $t$. By
definition of $\SPRval$ and by induction hypothesis, we have for
$s\notin\reloads$ that $b \geq \cons(s,a_{\min}) + \veccomp{\Safe}{t}$
and thus $\enlev{b}{sa_{\min}t}\geq \veccomp{\Safe}{t}$ and thus
$\run$ is $b$-safe by \Cref{lem:safe-path-extension}. If $s\in
\reloads$ and $b=0$, by similar arguments, as $\cons(s,a_{\min}) +
\veccomp{\Safe}{t} \leq \Ca$ (otherwise $b$ would be $\infty$), we
have that $\enlev{b}{sa_{\min}t}\geq \veccomp{\Safe}{t}$ and thus
$\run$ is $b$-safe.

Now we prove that $ b \leq d $. This clearly holds if $ d = \infty $,
so in the remainder of the proof we assume $ d\leq \Ca(\mdp) $. By the
definition of $ d $ there exists a strategy $ \sigma $ s.t. $\sigma $
is $ d $-safe in $ s $ and $\probm{\sigma}{s}(\ReachRuns^{i+1}(d))>0.$
Let $a=\sigma(a)$ be the action selected by $ \sigma $ in the first
step when starting in $s$. We denote by $ \tau $ the strategy such
that for all histories $ \fpath $ we have $ \tau(\fpath) =
\sigma(sa\fpath) $. For each $ t\in \Succ(s,a) $ we assign a number $
d_t $ defined as $d_t=0$ if $t\in\reloads$ and $d_t=\enlev{d}{sat}$
otherwise.

We finish the proof by proving these two claims:
\begin{compactenum}
\item It holds $ \SPRval(s, a, \mathcal{B}^{i}(\vec{y}_T)) \leq
\cons(s, a) + \max_{t\in \Succ(s,a)} d_t $.%
\item If $ s \not \in \reloads $, then $ \cons(s,a) + \max_{t\in
\Succ(s,a)} d_t \leq d$.
\end{compactenum}
Let us first see, why these claims are indeed sufficient. From (1.) we
get $ \veccomp{\mathcal{A}(\mathcal{B}^{i}(\vec{y}_T))}{s} \leq
\cons(s, a) + \max_{t\in \Succ(s,a)} d_t \leq \Ca(\mdp)$ (from the definition of $\enlev{d}{sat}$). If $ s\in
\reloads $, then it follows that $ 
\veccomp{\trunc{\mathcal{A}(\mathcal{B}^{i}(\vec{y}_T))}}{s} = 0 \leq
d $. If $ s\not \in \reloads $, then
$\veccomp{\trunc{\mathcal{A}(\mathcal{B}^{i}(\vec{y}_T))}}{s} = \veccomp{\mathcal{A}(\mathcal{B}^{i}(\vec{y}_T))}{s} \leq \cons(s,
 a) + \max_{t\in \Succ(s,a)} d_t \leq d$, the first inequality shown
 above and the second coming from (2.).

So let us start with proving (1.). Note that for each $ t\in
 \Succ(s,a) $, $ \tau $ is necessarily $ d_t $-safe from $ t $;
 hence, $ \safe(t) \leq d_t $. Moreover, there exists $ q \in
 \Succ(s,a) $ s.t. $ \probm{\tau}{q}(\ReachRuns^i)>0 $; hence,
 by induction hypothesis it holds $
 \veccomp{\mathcal{B}^{i}(\vec{y}_T)}{q} \leq d_q $. From this and
 from the definition of  $\SPRval$ we get
\begin{align*} 
\SPRval(s, a, \mathcal{B}^{i}(\vec{y}_T)) &
    \leq \cons(s,a) + \max \{ \veccomp{\mathcal{B}^{i}(\vec{y}_T)}{q}, 
    \safe(t) \mid t\in\Succ(s,a), t\neq q \}\\
&
    \leq \cons(s,a) + \max_{t\in \Succ(s,a)} d_t.
\end{align*}

To finish, (2.) follows immediately from the definition of $d_t$ and
$\enlev{d}{sat}$ as $\enlev{d}{sat}$ is always bounded from above by
$d-\Ca(s,a)$ for $s\notin \reloads$. \qed

\subsection{Proof of \Cref{lem:posreach-bound}}
\label{app:posreach-bound}
By \Cref{lem:posreach-iterate}, it suffices to show that $
\SafePosReach = \SafePosReach^{K}$. To this end, fix any state $ s $
such that $\SafePosReach(s) < \infty$. For the sake of succinctness,
we denote $\veccomp{\SafePosReach}{s} = d$. To any strategy $\pi$ that
is $d$-safe in $s$, we assign its \emph{index}, which is the infimum
of all $i$ such that $\probm{\pi}{s}(\ReachRuns^i) >0 $. By assumption
that $\SafePosReach(s) < \infty$, there is at least one $\pi$ with a
finite index. Let $\sigma$ be the $d$-safe in $s$ strategy with
minimal index $i$: we show that the $i \leq K$.

We proceed by a suitable ``strategy surgery'' on $\sigma$. Let
$\hist$ be a history produced by $\sigma$ from $s$ of length $i$
whose last state belongs to $T$. Assume, for the sake of
contradiction, that $i > K$. This can only be if at least one of the
following conditions hold:
\begin{compactenum}[(a)]
\item Some reload state is visited twice on $\hist$, i.e. there are $0
\leq j < k \leq i + 1$ such that $ \rstate[\hist]{k} =
\rstate[\hist]{j} \in \reloads $, or%
\item some state is visited twice with no intermediate visits to a
reload state; i.e., there are $ 0 \leq j < k \leq i + 1$ such that $
\rstate[\hist]{k} = \rstate[\hist]{j} $ and $ \rstate[\hist]{\ell}
\not\in \reloads $ for all $ j <\ell <k $.
\end{compactenum}
Indeed, if none of the conditions hold, then the reload states
partition $ \fpath $ into at most $ |\reloads|+1 $ segments, each
segment containing non-reload states without repetition. This would
imply $ i = \len{\hist} \leq K$.

In both cases (a) and (b) we can arrive at a contradiction using
essentially the same argument. Let us illustrate the details on case
(a): Consider a strategy $ \pi $ such that for every history of the
form $\hist\pref{j} \histconc\gamma $ for a suitable $ \gamma $ we
have $\pi(\hist\pref{j} \histconc\gamma) =
\sigma(\hist\pref{k}\histconc \gamma)$; on all other histories, $\pi$
mimics $\sigma$. Then $\pi$ is still $d$-safe in $s$. Indeed, the
behavior only changed on the suffixes of $\hist\pref{j}$, and for each
suitable $\gamma$ we have $\enlev{d}{\hist\pref{j}\histconc\gamma} =
\enlev{d}{\hist\pref{j}\histconc\gamma}$ due to the fact that
$\rstate[\hist]{j}=\rstate[\hist]{k}$ is a reload state. But $\pi$ has
index $i - (k-j)< i $, a contradiction with the choice of $\sigma$.

For case (b), the only difference is that now the resource level after
$\hist\pref{j}\histconc\gamma$ can be higher then the one of
$\hist\pref{k}\histconc\gamma$ due to the removal of the intermediate
non-reloading cycle. Since we need to show that the energy level never
drops below 0, the same argument works.%
\qed

\end{document}